%% file: qm-does-not-play-dice.tex
\documentclass[aps,superscriptaddress,nofootinbib,a4paper,longbibliography,twocolumn]{revtex4-2}
\usepackage{amsmath}
\usepackage{amsfonts}
\usepackage{amsthm}
\usepackage{amssymb}
\usepackage{graphicx}
\usepackage{enumerate}
\usepackage{color}
\usepackage[T1]{fontenc}
\usepackage{braket}
\usepackage{todonotes}
\usepackage{soul}
\usepackage{subfigure}
\usepackage{mathrsfs}
\usepackage{float}
\usepackage{amsmath}
\usepackage{bm}
\usepackage{multirow}
\usepackage{outlines}

\graphicspath{{figures/}}

\newtheorem{theorem}{Theorem}
\newtheorem{lemma}{Lemma}

\newtheorem*{theorem*}{Theorem}
\newtheorem{definition}{Definition}
\newtheorem{example}{Example}

\newtheoremstyle{operation}
  {\topsep}
  {\topsep}
  {}
  {}
  {\itshape}
  {.}
  {.5em}
  {\thmname{#1}\thmnumber{ #2}\thmnote{ (#3)}}
\theoremstyle{operation}
\newtheorem{operation}{Operation}

\def\>{\rangle}
\def\<{\langle}

\newcounter{protocol}
\renewcommand{\theprotocol}{\arabic{protocol}}

\newenvironment{protocol}[1]{%
  \par\addvspace{\topsep}%
  \refstepcounter{protocol}%
  \noindent\rule{\linewidth}{0.4pt}\par
  \vspace{0.2em}%
  \noindent\textbf{Protocol~\theprotocol.} #1\par
  \noindent\rule{\linewidth}{0.4pt}\par
  \ignorespaces
}{%
  \noindent\rule{\linewidth}{0.4pt}\par
  \addvspace{\topsep}%
  \ignorespacesafterend
}

\usepackage[bookmarks=false,colorlinks,citecolor=blue,
linkcolor=blue,anchorcolor=blue,urlcolor=blue
]{hyperref}

\begin{document}

\title{When Quantum Nonlocality Does Not Play Dice}
\author{Ravishankar Ramanathan}
\affiliation{School of Computing and Data Science, The University of Hong Kong, Pokfulam Road, Hong Kong}

\author{Yuan Liu}
\affiliation{School of Computing and Data Science, The University of Hong Kong, Pokfulam Road, Hong Kong}

\author{Stefano Pironio}
\affiliation{Laboratoire d'Information Quantique, Universit\'{e} libre de Bruxelles (ULB), Belgium}	
	
	
	
	\begin{abstract}
	Bell nonlocality is widely viewed as a signature of intrinsic randomness, effectively playing the role of a “dice” at the heart of many device-independent cryptographic protocols. We show that this connection has a fundamental limitation: there exist Bell inequalities that are maximally violated by quantum correlations yet certify no randomness for any fixed input pair. We develop a systematic construction based on symmetric deterministic extensions of nonlocal games, and use it to obtain explicit examples of such inequalities. We also construct maximally nonlocal quantum correlations that, for every input pair, admit a convex decomposition into strategies with predetermined outputs for those inputs. Our results reveal a strong form of determinism compatible with quantum nonlocality and delineate the limits of device-independent randomness certification.
\end{abstract}
	
	\maketitle
	

\textbf{Introduction.}
Nonlocality, witnessed by the violation of Bell inequalities, implies that the correlations $p = \{p(ab | xy)\} $ in a Bell experiment cannot be decomposed as a convex mixture of local deterministic models,
\begin{equation}\label{eq:lhv}
p(ab | xy) \;=\; \sum_{\lambda} q_{\lambda}\,\delta_{a,\alpha(x,\lambda)}\,\delta_{b,\beta(y,\lambda)},
\end{equation}  
where the outcomes $ a=\alpha(x,\lambda) $ and $ b=\beta(y,\lambda) $ are fully predetermined by the local settings $x$, $y$, and a hidden variable $\lambda$ \cite{Bell,RMPBellnonlocality}. Although nonlocality does not rule out more general deterministic models -- where outcomes also depend on distant measurement setting, i.e., $ a=\alpha(x,y,\lambda) $ and $ b=\beta(x,y,\lambda)$ -- such models violate no-signaling at the level of individual hidden variables, even if they respect it upon averaging over $\lambda$. Consequently, within any theory satisfying no-signaling -- such as quantum theory -- fully deterministic models must take the local form~\eqref{eq:lhv}, and are thus incompatible with nonlocality. The violation of a Bell inequality thus serves as a certificate of randomness. This carries not only profound foundational implications, but also practical utility: nonlocal correlations enable device-independent (DI) protocols for randomness certification \cite{PironioNature,KC,AM16} and secure quantum key distribution (QKD) \cite{BarrettPRL,DIQKD}, requiring no assumptions about the internal functioning of the devices.  

In this paper, we revisit this fundamental link between nonlocality and randomness, which is often taken for granted, and explore its limits. We show that there exist Bell inequalities that are maximally violated by quantum theory, yet fail to certify randomness for any given measurement pairs $(\hat x, \hat y)$. Going further, we show that there exist quantum correlations that are maximally nonlocal, yet from which no randomness can be certified, i.e., they do not violate any randomness-certifying Bell inequalities. for any choice of inputs $(\hat x,\hat y)$.

\textbf{Preliminaries.}
To state our results precisely, we first introduce some basic notions and draw a distinction between  the notions of a device-independent witness of randomness and a Bell inequality.

Consider a Bell experiment with correlations $p$ and let $(\hat x, \hat y)$ be a chosen pair of measurements for which we wish to certify randomness. This is, for instance, the central scenario in DI randomness generation and QKD protocols with a  spot-checking structure: random bits are extracted from a fixed input pair $(\hat x, \hat y)$ used in the majority of rounds; while other inputs are sampled less frequently to verify the Bell violation. This approach was introduced in \cite{AcinMassarPironio2006} for DIQKD and in \cite{PironioNature} for DIQRNG.

If the correlations take the full deterministic form \eqref{eq:lhv}, then no randomness can be inferred from the the outcomes of $(\hat x, \hat y)$, as from any other input pair. Bell inequality violation is therefore necessary for device-independent randomness certification. It is not, however, sufficient. Within quantum theory, one can consider a set $\mathcal{Q}_{\hat x \hat y}$ of correlations that are 
simultaneously nonlocal and deterministic for the specific input pair $(\hat{x}, \hat{y})$: 
  \begin{align}\label{eq:partly-deterministic-quantum}
  p \in \mathcal{Q}_{\hat x \hat y} \text{ if }& p = \sum_{\lambda} q_\lambda p_\lambda \text{ where} p_\lambda\in \mathcal{Q}\,,\nonumber \\
  &\text{and }p_\lambda(ab|\hat x \hat y) = \delta_{a,\alpha_\lambda}\delta_{b,\beta_\lambda}\,.
\end{align}
That is, $p$ decomposes into quantum strategies $p_\lambda$ in which the outcomes $a,b$ are fully determined by $\lambda$ for the input pair $(\hat x, \hat y)$, while outputs for other inputs may be random. The set $\mathcal{Q}_{\hat x \hat y}$ is intermediate between the local set $\mathcal{L}$ and the full quantum set $\mathcal{Q}$: $\mathcal{L}\subseteq \mathcal{Q}_{\hat x \hat y} \subseteq \mathcal{Q}$.

The sets $\mathcal{L}$, $\mathcal{Q}_{\hat x \hat y}$, and $\mathcal{Q}$ are all convex sets and thus fully characterized by linear inequalities. For a Bell expression $\omega[p] = \sum_{abxy} \omega_{abxy} p(ab|xy)$, let $L$, $Q_{\hat x\hat y}$, and $Q$ denote its maximum over $\mathcal{L}$, $\mathcal{Q}_{\hat x \hat y}$, and $\mathcal{Q}$, respectively. If $\omega[p]>{Q}_{\hat x \hat y}$, then $p \notin \mathcal{Q}_{\hat{x}\hat{y}}$, 
meaning the outcomes for $(\hat{x}, \hat{y})$ cannot have been predetermined within quantum 
theory: the inequality $\omega[p] \leq Q_{\hat{x}\hat{y}}$ thus acts as a witness of randomness 
for that input pair. Three regimes arise.
\emph{(i)} ${L} = {Q}_{\hat x \hat y}<{Q}$: any violation of the Bell inequality 
$\omega[p] \leq L$ certifies randomness for $(\hat{x}, \hat{y})$. The CHSH inequality 
is a paradigmatic example, as its violation certifies randomness for all input pairs. \emph{(ii)} ${L} < {Q}_{\hat x \hat y} < {Q}$: randomness is certified only if the violation 
exceeds $Q_{\hat{x}\hat{y}}$; violations in the range $]L, Q_{\hat{x}\hat{y}}]$ 
witness nonlocality but not randomness. \emph{(iii)} ${L} < {Q}_{\hat x \hat y} = {Q}$, the Bell inequality is never a witness of 
randomness for $(\hat{x}, \hat{y})$, since its maximal quantum violation can be achieved 
by correlations that are deterministic for that input pair.

An analogous construction applies if one assumes only the no-signaling principle, rather than the full constraints of quantum theory. One can define the set $\mathcal{N\!S}_{\hat{x}\hat{y}}$ of no-signaling correlations that are deterministic for $(\hat{x}, \hat{y})$ by replacing the condition $p_\lambda \in \mathcal{Q}$ in~\eqref{eq:partly-deterministic-quantum} with $p_\lambda \in \mathcal{N\!S}$, where $\mathcal{N\!S}$ is the set of no-signaling correlations. The inequalities $\omega[p] \leq N\!S_{\hat{x}\hat{y}}$, where $N\!S_{\hat{x}\hat{y}}$ is the maximum of $\beta$ over $\mathcal{N\!S}_{\hat{x}\hat{y}}$, then act as witnesses of randomness 
for $(\hat{x}, \hat{y})$ within the no-signaling framework. The concept of no-signaling yet partly deterministic correlations was first introduced in \cite{Woodhead14} and has since found applications in the context of the Local Friendliness no-go theorem \cite{Bong2020}; see \cite{Haddara2025} for further discussion and extensions. The quantum analogue 
$\mathcal{Q}_{\hat{x}\hat{y}}$ introduced here does not appear to have been previously considered.

Examples of inequalities satisfying \emph{(ii)} or \emph{(iii)} are easily constructed in both the quantum and no-signaling frameworks. Given any Bell scenario exhibiting nonlocality, one can extend it by appending an additional input pair $(\hat{x}, \hat{y})$ that is decoupled from the rest. Since the violation is entirely determined by the original inputs, the outcomes for  $(\hat{x}, \hat{y})$ can be made deterministic without affecting it. This shows that a Bell inequality violation is a strictly weaker requirement than device-independent randomness certification for a given input pair.

\textbf{Bell inequalities that fail to certify randomness.}
The above observations raise a much more fundamental question: do there exist Bell inequalities that are violated, yet fail to certify randomness for \emph{any} input pair? We term these \emph{ECID} Bell inequalities for "Every Chosen Input Deterministic" -- a name that also mirrors "DICE" in reverse, signifying their inability to certify randomness.  
Formally, for a given Bell expression $\omega[p]$, we define
\begin{equation}
  Q_{\mathrm{ECID}}  = \min_{\hat x, \hat y} Q_{\hat x \hat y}\,.
\end{equation}
The bound $Q_{\mathrm{ECID}}$ represents the threshold below which no randomness can be certified for any input pair: if $\omega[p]\leq Q_{\mathrm{ECID}}$, then none of the randomness witnesses $\omega[p]\leq Q_{\hat x \hat y}$ are violated. Conversely if $\omega[p]>Q_{\mathrm{ECID}}$, one of these witnesse is violated and randomness is certified for \emph{at least} one input pair. For most (and to our knowledge, all previously known) Bell expression $Q_{\mathrm{ECID}} = L$, meaning that any violation of the Bell inequality automotically certifies some randomness. This motivates the search for Bell expressions where $Q_{\mathrm{ECID}} > L$. Going further, we can ask if there exist Bell inequalities such that 
\begin{equation}
  Q_{\mathrm{ECID}} = Q\,.
\end{equation}
For such an inequality, the threshold for certifying randomness coincides with the maximum quantum violation. Thus no randomness can be certified for any input pair, regardless how strong the inequality is violated by quantum theory. We say a Bell inequality with the above property is a \emph{full ECID} inequality. From a practical standpoint, full ECID inequalities are useless for spot-checking DI protocols.

Naively, one might expect ECID inequalities to be impossible. Since a Bell violation rules out models that are deterministic for \emph{all} inputs, it seems intuitive that at least one input pair $(\hat{x}, \hat{y})$ must necessarily certify randomness when the Bell inequality is violated, i.e., that $Q_{\hat x \hat y} = L$ for some $(\hat{x}, \hat{y})$ and thus $Q_{\mathrm{ECID}} = L$. However, this is not necessarily the case, since the existence of several quantum strategies each deterministic for a single input pair $(\hat{x}, \hat{y})$ does not necessarily imply the existence of a \emph{single} strategy of the form \eqref{eq:lhv} where all inputs are \emph{simultaneously} deterministic.

A parallel question arises in the no-signaling framework: are there Bell inequalities such that
\begin{equation}\label{eq:ecid-no-signaling}
N\!S_{\mathrm{ECID}} = \min_{\hat x, \hat y} N\!S_{\hat x \hat y} > L\,,
\end{equation}
and even further, such that $N\!S_{\mathrm{ECID}} = N\!S$,
meaning the maximal no-signaling violation can be achived by strategies that are partly detrministic for any inputs. 

To our knowledge, quantum ECID Bell inequalities have not been previously reported. In contrast, a known example exists in the no-signaling framework: the inequality associated with the Magic Square (MS) game \cite{ACPPS16}. The construction of no-signaling ECID inequalities is generally easier because $\mathcal{N\!S}_{\hat{x}\hat{y}}$ is a larger set than $\mathcal{Q}_{\hat{x}\hat{y}}$: there are more no-signaling strategies available to construct correlations both nonlocal and partly deterministic. Indeed, while the MS inequality is full ECID in the no-signaling case, its maximal violation certifies randomness for all inputs according to quantum theory.

Our first main result is to prove that quantum full ECID Bell inequalities exist. To establish this, we introduce a systematic construction termed symmetric deterministic extension (SDE), which allows us, starting from certain nonlocal games, to build ECID inequalities, both in the quantum and no-signaling frameworks.

We illustrate this procedure starting from the CHSH game \cite{CHSH69}, which we reformulate as the nonlocal game $G_2$ defined by the $2\times 2$ table in Fig.~\ref{fig:games}. Alice and Bob receive inputs corresponding to a row $x\in\{1,2\}$ and a column $y\in\{1,2\}$, respectively. Alice assigns binary values $v_{x1},v_{x2}\in\{\pm 1\}$ to her row subject to the parity constraints $v_{x1} \cdot v_{x2} = 1$. Bob assigns binary values $v'_{1y},v'_{2y}\in\{\pm 1\}$ to his column, subject to the parity constraints $v'_{1y} \cdot v'_{2y} = \gamma_y$, where $\gamma_1 = 1$ and $\gamma_2 = -1$. 
They win if their assignments agree at the intersection entry, i.e., $v_{xy} = v'_{xy}$.
The average winning probability
\begin{equation}\label{eq:winchsh}
    \omega_2[p] = \frac{1}{4}\sum_{x,y} p(v_{xy} = v'_{xy} | x,y)\,,
\end{equation}
defines the Bell expression for the game. This formulation is equivalent to the standard CHSH game: by identifying Alice's output as $a \equiv v_{x1}=v_{x2}$ and Bob's output as $b\equiv v'_{1y}=v'_{2y}\cdot\gamma_y$, the winning condition becomes the standard CHSH condition $a\cdot b = (-1)^{\delta_{x,2}\delta_{y,2}}$. The optimal winning probabilities are $L^{(2)} = 3/4$, $Q^{(2)} = (2+\sqrt{2})/4$, and $N\!S^{(2)} = 1$. Note that any variant of this game obtained by altering the parity constraints remains equivalent to $G_2$ (up to relabeling of inputs and outputs), provided the total number of odd-parity constraints (dashed lines) is odd (see Appendix~A).

\begin{figure}[t]
\centering
\begin{tikzpicture}[scale=1, every node/.style={font=\small}]
  \begin{scope}[shift={(-1,0)}]
    \def\n{2}
    \def\s{0.9} 

    \foreach \r in {1,...,\n} {
      \draw[line width=1pt]
        ({0.5*\s},{(\r-0.5)*\s}) -- ({(\n-0.5)*\s},{(\r-0.5)*\s});
    }

    \foreach \c in {1,...,\numexpr\n-1\relax} {
      \draw[line width=1pt]
        ({(\c-0.5)*\s},{0.5*\s}) -- ({(\c-0.5)*\s},{(\n-0.5)*\s});
    }
    \draw[line width=1pt, dashed]
      ({(\n-0.5)*\s},{0.5*\s}) -- ({(\n-0.5)*\s},{(\n-0.5)*\s});

    \foreach \r in {1,...,\n} {
      \foreach \c in {1,...,\n} {
        \fill ({(\c-0.5)*\s},{(\r-0.5)*\s}) circle (1.3pt);
        \pgfmathtruncatemacro{\rtop}{\n-\r+1}
        \node[anchor=south east]
          at ({(\c-0.24)*\s - 0.12},{(\r-0.62)*\s + 0.12}) {$v_{\rtop\c}$};
      }
    }
  \end{scope}

  \begin{scope}[shift={(1,-0.5)}]
    \def\n{3}
    \def\s{0.9}

    \foreach \r in {1,...,\n} {
      \draw[line width=1pt]
        ({0.5*\s},{(\r-0.5)*\s}) -- ({(\n-0.5)*\s},{(\r-0.5)*\s});
    }

    \foreach \c in {1,...,\numexpr\n-1\relax} {
      \draw[line width=1pt]
        ({(\c-0.5)*\s},{0.5*\s}) -- ({(\c-0.5)*\s},{(\n-0.5)*\s});
    }
    \draw[line width=1pt, dashed]
      ({(\n-0.5)*\s},{0.5*\s}) -- ({(\n-0.5)*\s},{(\n-0.5)*\s});

    \foreach \r in {1,...,\n} {
      \foreach \c in {1,...,\n} {
        \fill ({(\c-0.5)*\s},{(\r-0.5)*\s}) circle (1.3pt);
        \pgfmathtruncatemacro{\rtop}{\n-\r+1}
        \node[anchor=south east]
          at ({(\c-0.24)*\s - 0.12},{(\r-0.62)*\s + 0.12}) {$v_{\rtop\c}$};
      }
    }
  \end{scope}

  \begin{scope}[shift={(4.,-1)}]
    \def\n{4}
    \def\s{0.9}

    \foreach \r in {1,...,\n} {
      \draw[line width=1pt]
        ({0.5*\s},{(\r-0.5)*\s}) -- ({(\n-0.5)*\s},{(\r-0.5)*\s});
    }

    \foreach \c in {1,...,\numexpr\n-1\relax} {
      \draw[line width=1pt]
        ({(\c-0.5)*\s},{0.5*\s}) -- ({(\c-0.5)*\s},{(\n-0.5)*\s});
    }
    \draw[line width=1pt, dashed]
      ({(\n-0.5)*\s},{0.5*\s}) -- ({(\n-0.5)*\s},{(\n-0.5)*\s});

    \foreach \r in {1,...,\n} {
      \foreach \c in {1,...,\n} {
        \fill ({(\c-0.5)*\s},{(\r-0.5)*\s}) circle (1.3pt);
        \pgfmathtruncatemacro{\rtop}{\n-\r+1}
        \node[anchor=south east]
          at ({(\c-0.24)*\s - 0.12},{(\r-0.62)*\s + 0.12}) {$v_{\rtop\c}$};
      }
    }
  \end{scope}
\end{tikzpicture}

\caption{From left to right: the CHSH game \(G_2\), the Magic Square game \(G_3\), and the  game \(G_4\). Dot marks entries $v_{xy}$ for Alice and $v'_{xy}$ for Bob; rows corresponds to Alice's inputs, columns to Bob's. Solid (dashed) lines indicate even ($+1)$ [odd ($-1$)] parity constraints on the corresponding row or column.
The winning condition is $v_{xy} = v'_{xy}$.}
\label{fig:games}
\end{figure}

For the CHSH game, $Q^{(2)}_{\hat x \hat y} = N\!S^{(2)}_{\hat x \hat y} = L$ for all $(\hat x, \hat y)$, i.e., a violation of the CHSH inequality certifies randomness for any inputs.
We now apply the SDE procedure to $G_2$ to construct a game that fails to certify randomness for any input in the no-signaling framework.
This is done by adding one row and one column to the CHSH table, yielding the $3\times 3$ table shown in Fig.~\ref{fig:games}.
The parity constraints remain the same: each of Alice's row $(v_{x1}, v_{x2}, v_{x3})$ must have even parity (solid lines), and each of Bob's column $(v'_{1y}, v'_{2y}, v'_{3y})$ must have even parity, except for the last colum which has odd parity (dashed vertical line). The winning condition remains $v_{xy} = v'_{xy}$ and the average winning probability is
\begin{equation}
\omega_3[p] = \frac{1}{9} \sum_{x,y} P(v_{xy} = v'_{xy}| x,y)\,.
\end{equation}
This game turns out to be the well-known MS game \cite{Mermin1990,Aravind2002}, with $L^{(3)} =8/9$, $Q^{(3)} = N\!S^{(3)} = 1$. 
The perfect quantum strategy $Q^{(3)}$ is realized via products of Pauli measurements on two copies of a maximally entangled two-qubit state.

The MS game extension of the CHSH game is \emph{deterministic} in the following sense: if Alice and Bob fix deterministic values $v_{\hat x,y}=\lambda(\hat x,y)$ for the additional row $\hat x$ and $v'_{x,\hat y}=\lambda(x,\hat y)$ for the additional column $\hat y$ that satisfy the parity constraints and agree on the intersection entry, the remaining $2\times 2$ table is equivalent to the original CHSH game (see example in Fig.~\ref{fig:reduction} and formal proof in Appendix~A). Furthermore, this extension is \emph{symmetric}, meaning this reduction to the original CHSH game holds for any chosen input pair $(\hat x, \hat y)$.

\begin{figure}[t]
\centering
\begin{tikzpicture}[scale=1, every node/.style={font=\small}]
  \def\s{0.9}

  \begin{scope}[shift={(0,0)}]
    \foreach \r in {1,...,3} {
      \draw[line width=1pt] ({0.5*\s},{(\r-0.5)*\s}) -- ({2.5*\s},{(\r-0.5)*\s});
    }
    \foreach \c in {1,2} {
      \draw[line width=1pt] ({(\c-0.5)*\s},{0.5*\s}) -- ({(\c-0.5)*\s},{2.5*\s});
    }
    \draw[line width=1pt, dashed] ({2.5*\s},{0.5*\s}) -- ({2.5*\s},{2.5*\s});

    \foreach \r in {1,...,3} {
      \foreach \c in {1,...,3} {
        \fill ({(\c-0.5)*\s},{(\r-0.5)*\s}) circle (1.4pt);
        \pgfmathtruncatemacro{\rtop}{3-\r+1}
        \node[anchor=south east] at ({(\c-0.5)*\s+0.1},{(\r-0.5)*\s+0.}) {$v_{\rtop\c}$};
      }
    }
  \end{scope}

  \draw[->, line width=1.1pt] ({3.1}, {1.35}) -- ({4.2}, {1.35});

  \begin{scope}[shift={(4.8,0)}]
    \foreach \r in {1,...,3} {
      \foreach \c in {1,...,3} {
        \fill ({(\c-0.5)*\s},{(\r-0.5)*\s}) circle (1.4pt);
      }
    }

    \draw[line width=1pt] ({0.5*\s},{2.5*\s}) -- ({2.5*\s},{2.5*\s});

    \draw[line width=1pt] ({1.5*\s},{0.5*\s}) -- ({1.5*\s},{2.5*\s});

    \draw[line width=1pt, red, dashed] ({0.5*\s},{1.5*\s}) -- ({1.35*\s},{1.5*\s});
    \draw[line width=1pt, red, dashed] ({1.65*\s},{1.5*\s}) -- ({2.5*\s},{1.5*\s});
    \draw[line width=1pt, red] ({0.5*\s},{0.5*\s}) -- ({1.35*\s},{0.5*\s});
    \draw[line width=1pt, red] ({1.65*\s},{0.5*\s}) -- ({2.5*\s},{0.5*\s});
    \draw[line width=1pt, red] ({0.5*\s},{0.5*\s}) -- ({0.5*\s},{1.5*\s});
    \draw[line width=1pt, red] ({2.5*\s},{0.5*\s}) -- ({2.5*\s},{1.5*\s});

    \node[anchor=south east] at ({0.5*\s+0.1},{2.5*\s+0.}) {$+1$};
    \node[anchor=south east] at ({1.5*\s+0.1},{2.5*\s+0.}) {$-1$};
    \node[anchor=south east] at ({2.5*\s+0.1},{2.5*\s+0.}) {$-1$};
    \node[anchor=south east] at ({0.5*\s+0.10},{1.5*\s+0.}) {$v_{21}$};
    \node[anchor=south east] at ({1.5*\s+0.10},{1.5*\s+0.}) {$-1$};
    \node[anchor=south east] at ({2.5*\s+0.10},{1.5*\s+0.}) {$v_{23}$};
    \node[anchor=south east] at ({0.5*\s+0.10},{0.5*\s+0.}) {$v_{31}$};
    \node[anchor=south east] at ({1.5*\s+0.10},{0.5*\s+0.}) {$+1$};
    \node[anchor=south east] at ({2.5*\s+0.10},{0.5*\s+0.}) {$v_{33}$};
  \end{scope}
\end{tikzpicture}
\caption{Reduction of the $G_3$ Magic Square table (left) to a $G_2$ subtable (right) by fixing the values of one row and one column. To preserve the original parity constraints, the $2\times 2$ subtable must retain the parity lines shown in red. Since this subtable contains an odd number of dashed lines, the resulting game is equivalent to the CHSH game $G_2$. This holds for any valid deterministic assignment to the fixed row and column.}
\label{fig:reduction}
\end{figure}

Using the optimal no-signaling strategy for the original CHSH game, achiving $N\!S^{(2)} = 1$, to determine the entries for the $2\times 2$ subtable yields a winning probability for the $G_3$ game of $\omega_3[p]=1$ while maintaining deterministic values for the fixed row $\hat x$ and column $\hat y$. Thus $N\!S^{(3)}_{\hat x \hat y} = N\!S^{(3)}=1$ and this is true for any choice of $(\hat x, \hat y)$. Thus $N\!S^{(3)}_{\mathrm{ECID}} = N\!S^{(3)}$ and the MS expression is a no-signaling full ECID inequality. While this fact was previously noted in \cite{ACPPS16}, we established it here via an explicit constructive argument.

In contrast, because quantum strategies can satisfy the $2\times 2$ subtable only with a winning probability $Q^{(2)} = \sqrt{2}/2<1$, the same row-column fixing construction yields quantum strategies for $G_3$ with $\omega[p] = (5 + 2\sqrt{2})/9\approx 0.870> L^{(3)} = 8/9$. Consequently, any violation in the range $]8/9, 0.870]$ fails to certify randomness for any input pair in the quantum framework. While $Q^{(3)}_{\mathrm{ECID}}>L$ and $G_3$ is thus a quantum ECID inequality, it is not a full one, because the maximal quantum violation $Q^{(3)} = 1$ exceeds this threshold.

To obtain a quantum full ECID inequality, we iterate the SDE procedure. By extending $G_3$ to a $4\times 4$ game $G_4$, as shown in Fig.~\ref{fig:games}, we obtain a new game with $L^{(4)} = 15/16$. The Bell inequality $\omega_4[p]\leq 15/16$ is a facet of the Bell polytope. 
As in the previous case, any deterministic assignment to a row $\hat x$ and column $\hat y$ leaves a $3\times 3$ subtable equivalent to the MS game $G_3$. Since $Q^{(3)} = 1$, the perfect quantum strategy for $G_3$ can be used to determine values for the $3\times 3$ subtable, yielding a perfect quantum strategy for $G_4$ that is deterministic for any input pair $(\hat x, \hat y)$, implying $Q^{(4)}_{\hat x \hat y} = Q^{(4)} = 1$. Therefore, $G_4$ provides an example of a quantum full ECID inequality for which $Q^{(4)}_{\mathrm{ECID}} = Q$: even at the maximal quantum violation, no randomness is certified for any input pair. Since quantum correlations are a subset of no-signaling correlations, this is also a no-signaling full ECID inequality. We note that geometrically, quantum full ECID inequalities correspond to supporting hyperplanes of the quantum set $\mathcal{Q}$ that are highly dimensional and correspond to flat regions of its boundary (since there exists several quantum strategies that achieve the same maximal violation).

Our general SDE construction relies on three ingredients. Given a nonlocal game $G$ with a perfect winning strategy: \emph{(i)} we construct an extended game $G'$ with additional inputs and outputs; \emph{(ii)} the extension is designed such that fixing deterministic values for the additional inputs and outputs recovers the original game $G$ (possibly up to relabeling of inputs and outputs and the addition of dummy inputs and outputs); and \emph{(iii)} the construction is symmetric, ensuring the reduction to $G$ holds for any any chosen input pair $(\hat x, \hat y)$.
When these conditions are satisfied, any perfect quantum (or no-signaling) winning strategy for $G$ can be extended to a perfect quantum (or no-signaling) winning strategy for $G'$ that is deterministic for any specified input pair $(\hat x, \hat y)$.
Appendix~B describes how such SDE constructions can be applied to a broad   family of nonlocal games, including the Mermin star game \cite{Mermin1990} and the GHZ game \cite{GHZ89,GHSZ90}.

\textbf{Nonlocal correlations without certifiable randomness.}
The existence of ECID Bell inequalities raises a new question:  are there nonlocal quantum correlations from which no randomness can be certified for any inputs $(\hat x, \hat y)$, regardless of the Bell inequality used?

We define a quantum correlation $p$ as ECID if it belongs to the intersection 
\begin{equation}
  \mathcal{Q}_{\mathrm{ECID}} = \bigcap_{\hat x, \hat y} \mathcal{Q}_{\hat x \hat y}
\end{equation}
of all partly deterministic sets. A correlation $p\in \mathcal{Q_{\mathrm{ECID}}}$ thus admits, for every possible input pair $(\hat x, \hat y)$, a convex decomposition $p = \sum_{\lambda} q_\lambda p_\lambda$ where the outcomes for $(\hat x, \hat y)$ are fully predetermined by $\lambda$ (the  convex decomposition can depend on the chosen input pair). Thus, for any fixed $(x,y)$, one cannot exclude that $p$ was prepared using the decomposition in which $\lambda$ fully determines the outputs -- making 
it impossible to certify any intrinsic randomness. ECID correlations are therefore useless for device-independent spot-checking protocols.

Note that the existence of nonlocal but ECID correlations does not follow automatically from the existence of ECID inequalities. An ECID inequality $\omega$ is a nonlocal witness that such that for any input $(\hat x, \hat y)$, there exists a strategy $p^{\hat x \hat y}\in \mathcal{Q}_{\hat x\hat y}$ that is both deterministic on $(\hat x, \hat y)$ and violates the inequality. However, each of these strategies may be random for all inputs other than $(\hat x, \hat y)$ and thus may not belong to $\mathcal{Q}_{x y}$ for $(x,y) \ne (\hat x, \hat y)$. For instance, each of the partly deterministic strategies $p^{\hat x \hat y}$ constructed above and achieving $Q_{\hat x \hat y}^{(4)} = 1$ for $G_4$ yields intrisinsically random outputs for other pairs $(x,y) \ne (\hat x, \hat y)$. 
Thus, while ECID inequalities guarantee the existence of different correlations $p^{\hat x \hat y}\in \mathcal{Q}_{\hat x\hat y}$ that violate the same inequality, they do not guarantee the existence of a \emph{single} correlation $p$ that belongs to each set $\mathcal{Q}_{\hat x \hat y}$ and thus to $\mathcal{Q}_{\mathrm{ECID}}$.

More formally, given a Bell expression $\omega[p]$, define $Q_{\mathrm{ECID}^*} = \max_{p\in \mathcal{Q}_{\mathrm{ECID}}} \omega[p]$ as the maximum value achievable by an ECID correlation. Since $\mathcal{Q}_{\hat x \hat y}$ is a superset of $\mathcal{Q}_{\mathrm{ECID}}$ for each $(\hat x, \hat y)$, we have $Q_{\mathrm{ECID}^*} \leq Q_{\mathrm{ECID}}$. However, the existence of a violated ECID inequality (i.e., $Q_{\mathrm{ECID}} > L$) does not necessarily imply the existence of nonlocal ECID correlations (i.e., $Q_{\mathrm{ECID}^*} > L$).

Similarly, we define no-signaling ECID correlations as those in $\mathcal{N\!S}_{\mathrm{ECID}} = \bigcap_{\hat x, \hat y} \mathcal{N\!S}_{\hat x \hat y}$. Examples of quantum and no-signaling ECID correlations were given in \cite{ACPPS16}. However, they relied on detection inefficiencies and the constructed correlations had a strictly positive local fraction, i.e, they were mixtures of nonlocal and local correlations.
Here, we provide explicit examples of \emph{maximally nonlocal} ECID correlations—i.e., with zero local fraction, or equivalently,, that lie on the boundary of the quantum or no-signaling sets. These examples show that ECID correlations are not mere artifacts of noise or detection loopholes, but an intrinsic feature of certain nonlocal correlations.

Our construction builds on the strategies for $G_3$ and $G_4$, described above. 
For each fixed $(\hat x, \hat y)$, these strategies assign deterministic values $v_{\hat x,y}=\lambda(\hat x,y)$ and $v'_{x,\hat y}=\lambda(x,\hat y)$ to the outputs of the chosen row and column.  
Every valid $\lambda$ yields a correlation $p_\lambda^{(\hat x,\hat y)} \in \mathcal{Q}_{\hat x \hat y}$ that wins $G_3$ or $G_4$ with certainty.  The correlations $p_\lambda^{(\hat x,\hat y)}$ are generally different for different $\lambda$ and different choices of $(\hat x, \hat y)$. If we take the uniform mixture over all valid  $\lambda$: we get some correlations $p^{(\hat x,\hat y)} = \frac{1}{|\Lambda|}\sum_\lambda p_\lambda^{(\hat x,\hat y)}$ that still belong to $\mathcal{Q}_{\hat x \hat y}$ and still win $G_3$ or $G_4$ perfectly. Remarkably, by taking this uniform mixture, it happens that the dependency on the choice of $(\hat x, \hat y)$ disappears: $p^{\hat x \hat y}= p$ and we obtain a single quantum correlation in $\mathcal{Q}_{\mathrm{ECID}}$ that still win $G_3$ or $G_4$ perfectly, and which is thus maximally nonlocal. Details are given in Appendix~C.

For $G_3$, the resulting $p$ is no-signaling ECID and in fact coincides with the optimal quantum strategy for the MS game. For $G_4$, the construction yields the first example of maximally nonlocal \emph{quantum} ECID correlations.
Formally, this also implies that $N\!S^{(3)}_{\mathrm{ECID}^*} = N\!S^{(3)}_{\mathrm{ECID}}=N\!S^{(3)}$ and $Q^{(4)}_{\mathrm{ECID}^*} = Q^{(4)}_{\mathrm{ECID}} = Q^{(4)}$. 

\textbf{Discussion.}  
  We have identified explicit ``ECID'' Bell inequalities and quantum correlations for which nonlocality fails to certify randomness for any fixed inputs, and constructed maximally nonlocal examples of them.

  Strategies exploiting detection inefficiencies in \cite{ACPPS16,Masini} implicitly provide examples of quantum and no-signaling ECID correlations, but these are not maximally nonlocal and only happen in the presence of ``no-click'' events. After a first version of the present work appeared on-line, Ref.~\cite{tura} showed that a 4-input Bell inequality is partly ECID in the quantum framework (the violation needs to be above a certain threshold to certify randomness). Prior to the present work, however, there were no examples of full quantum ECID inequalities, nor of maximally nonlocal quantum ECID correlations. In the no-signaling case, the only known example was based on the MS game. Here, we provide examples of full quantum ECID inequalities and maximally nonlocal quantum ECID correlations, together with a general construction that applies to a broad family of nonlocal games, yielding additional examples in both the quantum and no-signaling frameworks.

  Conceptually, ECID inequalities and correlations show that a strong form of determinism can coexist with nonlocality. Operationally, this means that an adversary Eve who knows in advance which input pair $(\hat x, \hat y)$ will be used to generate randomness can prepare the observed correlations as a convex decomposition giving her full information about the corresponding outputs. 

   A stronger notion of predictability arises if Eve learns the inputs only \emph{after} preparing the correlations and must then guess the outputs.  Nonlocal correlations whose outcomes remain fully predictable in that scenario were termed to exhibit \emph{bound randomness} in \cite{ACPPS16}. In the no-signalling framework, the existence of nonlocal ECID correlations directly implies the existence of bound randomness \cite{ACPPS16,RLWP25}. The quantum case is sharply different: as shown in \cite{RLWP25}, bound randomness does not exist in quantum theory. Quantum ECID correlations thus realize one of the strongest forms of determinism still compatible with quantum nonlocality.

  Besides clarifying the subtle relationship between nonlocality and randomness, our results also point to a rich geometry. ECID inequalities are inequalities that are valid for all sets $\mathcal{Q}_{\hat x \hat y}$ or $\mathcal{N\!S}_{\hat x \hat y}$, that is they correspond to  valid inequalities for the union of these sets. On the other hand, ECID correlations are points that lie in the intersection of all such sets. Understanding ECID phenomena therefore requires understanding the sets \( \mathcal{Q}_{\hat x\hat y} \) and \( \mathcal{N\!S}_{\hat x\hat y} \) individually, as well as the structure of their union and intersection. While the partly deterministic no-signalling sets \( \mathcal{N\!S}_{\hat x\hat y} \) were introduced previously and some of their properties have been investigated \cite{Woodhead14,Haddara2025}, the quantum sets $\mathcal{Q}_{\hat x\hat y}$, which are more relevant for quantum applications, remain essentially unexplored.

 Finally, while we have focused on models that are deterministic for a joint input pair \( (\hat x,\hat y) \), one can also define a weaker, one-sided notion in which determinism is required only for one party, say for Alice's fixed input \( \hat x \). Our results cover this setting as well. In general, however, one-sided ECID inequalities and correlations are strictly weaker than the two-sided ones, since two-sided determinism implies one-sided determinism, but not conversely. Exploring this weaker notion may help further clarify the relationship between randomness and quantum nonlocality, in particular for device-idependent quantum key distribution.

\textit{Acknowledgments.-}
We acknowledge usefull discussions with Zhou Yangchen, Shuai Zhao and Ho Yiu Chung. R.R. acknowledges support from the General Research Fund (GRF) Grant No.\ 17211122, and the Research Impact Fund (RIF) Grant No.\ R7035-21.	
S.P. acknowledges funding  from the VERIqTAS project within the QuantERA II Programme that has received funding from the European Union's Horizon 2020 research and innovation program under
Grant Agreement No 101017733 and the F.R.S-FNRS Pint-Multi program under Grant Agreement R.8014.21, from the European Union’s Horizon Europe research and innovation programme under the project "Quantum Security Networks Partnership" (QSNP, grant agreement No 101114043), 
from the F.R.S-FNRS through the PDR T.0171.22,
from the FWO and F.R.S.-FNRS under the Excellence of Science (EOS) programme project 40007526,
from the FWO through the BeQuNet SBO project S008323N,from the Belgian Federal Science Policy through the contract RT/22/BE-QCI and the EU “BE-QCI” program. 
S.P. is a Research Director of the Fonds de la Recherche Scientifique – FNRS.

Funded by the European Union. Views and opinions expressed are however those of the authors only and do not necessarily reflect those of the European Union. Neither the European Union nor the granting authority can be held responsible for them.

  \bibliographystyle{apsrev4-2}
  \bibliography{references}

	\widetext
	\newpage
	\appendix
	\setcounter{table}{0}
	\renewcommand{\thetable}{A\arabic{table}}
	\setcounter{figure}{0}
	\renewcommand{\thefigure}{A\arabic{figure}}
	\setcounter{definition}{0}
	\renewcommand{\thedefinition}{A\arabic{definition}}
	\setcounter{theorem}{0}
	\renewcommand{\thetheorem}{A\arabic{theorem}}
	\setcounter{lemma}{0}
	\renewcommand{\thelemma}{A\arabic{lemma}}
	\setcounter{proposition}{0}
	\renewcommand{\theproposition}{A\arabic{proposition}}
	\setcounter{example}{0}
	\renewcommand{\theexample}{A\arabic{example}}
	\setcounter{corollary}{0}
	\renewcommand{\thecorollary}{A\arabic{corollary}}

	\section*{Appendix}

  \input{SM}

\end{document}

%% file: SM.tex





\setcounter{operation}{0}
\renewcommand{\theoperation}{A\arabic{operation}}





\setcounter{protocol}{0}
\renewcommand{\theprotocol}{A\arabic{protocol}}


	
	





This Appendix provides the technical details supporting the results announced in the main text.
Appendix~\ref{BLCS} introduces Binary Linear Constraint Systems (BLCSs) and establishes two auxiliary lemmas that are needed in the proof of the main construction: a sufficient condition for a BLCS to have no classical solution, and a standard-form reduction. This algebraic background is necessary because the nonlocal games used in Appendix~\ref{app:SDE} are instances of BLCS-games, and the Symmetric Deterministic Extension (SDE) protocol operates directly on the BLCS structure (the game $G_2$ and its extensions in the main text are instances of BLCS-games).
Appendix~\ref{app:SDE} formalises the SDE construction introduced informally in the main text. It gives the precise definition of an SDE of a nonlocal game, proves that every BLCS-game with a perfect quantum but no perfect local strategy admits an SDE, and illustrates on the Mermin-GHZ game that the construction also works in the multipartite case.
Appendix~\ref{subsec:SDLMSquare} constructs the explicit nonlocal quantum ECID correlation claimed in the main text. It shows that this single correlation, which wins the SDE game $G_4$ with unit probability, admits for every input pair $(\hat x, \hat y)$ a convex decomposition into perfect quantum strategies that are deterministic on $(\hat x, \hat y)$, thereby establishing that no randomness can be extracted from it for any choice of inputs.

\section{Binary linear constraint systems} \label{BLCS}
A Binary Linear Constraint System (BLCS) $A$ consists of $p$ binary variables $V(A)=\{v_1,\ldots,v_p\}$ and $q$ linear constraint equations $C(A)=\{1,\ldots, q\}$ \cite{CM12, Ji13}. Each of the constraint equations is a binary-valued linear function of some subset of the $p$ variables. We define the binary variables $v_i$ over $\{+1, -1\}$ in which case the $j$-th constraint equation is in general written as
\begin{equation}
		\prod_{i=1}^{p}  v_i^{\alpha_{i,j}}=\beta_j.
\end{equation}
Here $\alpha_{i,j} \in \{0,1\}$ indicates whether the binary variable $v_i$ (for $i=1,\ldots, p$) appears ($\alpha_{i,j}=1$) or does not appear ($\alpha_{i,j}=0$) in the $j$-th constraint equation. The parameter $\beta_j \in \{+1, -1 \}$ is said to be the \textit{parity} of the $j$-th constraint equation. The parity of the whole binary linear constraint system $A$ is defined as 
\begin{equation}
\text{par}(A) = \prod_{j=1}^q \beta_j \in \{+1,-1\}.
\end{equation}
The \textit{degree} of the binary variable $v_i$ is defined as the number of constraint equations it appears in, i.e.,
\begin{equation}
\text{deg}(v_i) = \sum_{j=1}^q \alpha_{i,j}.
\end{equation}

A \emph{classical solution} to the above BLCS is an assignment of deterministic values in $\{+1, -1\}$ to each of the binary variables $v_i$ such that all of the $q$ constraint equations are satisfied. It is a known hard problem to find a classical solution or to determine the classical realizability of a BLCS. 

On the other hand, a \emph{quantum solution} to the above BLCS associates with a (finite-dimensional) state $|\psi\>$ and a set of (finite-dimensional) Hermitian operators $\{O_i\}$ with outcomes being $\{\pm 1\}$ (wherein each of the variables $v_i$ is associated with a quantum binary observable $O_i$) satisfying that
	\begin{itemize}
		\item [1.] For any two binary variables $v_i,v_j\in V(A)$ that appear in the same constraint equation, the corresponding binary observables $O_i,O_j$ commute.
		\item [2.] For any constraint equation it holds that
        \begin{equation}
	   \prod_{i=1}^{p}  O_i^{\alpha_{i,j}}|\psi\>=\beta_j|\psi\>, \; \forall j\in C(A).
        \end{equation} 

	\end{itemize}
Note that if there exists a set of (finite-dimensional) Hermitian operators $\{O_i\}$ with outcomes in $\{\pm 1\}$ that satisfy the above conditions not only for a specific state $|\psi\>$ but for any state in that dimension, then this is referred to as an \emph{operator solution} of the BLCS.

When two binary variables $v_i$ and $v_{k}$ appear in the same constraint equation in a BLCS, we denote it by $v_i \sim v_{k}$. Denote by $V(A)$ and $V(B)$ the sets of binary variables that appear in two BLCS $A$ and $B$. We say that the two BLCS $A$ and $B$ are homomorphic if there exists a function $f: V(A) \rightarrow V(B)$ such that $f(v_i) \sim f(v_k)$ in $B$ whenever $v_i \sim v_k$ in $A$. 
We say that two BLCSs $A$ and $B$ are isomorphic if there exists Bijective function $f: V(A) \rightarrow V(B)$, such that system $A$ on the variables $f(V(A))$ is identical to system $B$, up to a sign flip applied to the variables. It is clear that two isomorphic BLCSs have the same classical and quantum realizability. 

We observe the following simple sufficient condition for a BLCS to admit no classical solution. 
\begin{lemma}~\label{obs}
Consider a binary linear constriant system $A$ in which the degree of each binary variable is even (i.e., $\text{deg}(v_i)$ is even for all $v_i \in V(A)$) and the parity of the system is $-1$ (i.e., $\text{par}(A) = -1$). Then the system $A$ does not admit a classical solution. 

\end{lemma}
	\begin{proof}
Consider the following equation obtained by multiplying all of the $q$ constraint equations in the BLCS
	\begin{equation}\label{eq_obs}
		\prod_{j=1}^{q}\left(\prod_{i=1}^{p}  v_i^{\alpha_{i,j}}\right)=\prod_{j=1}^{q} \beta_j,
	\end{equation}
	Observe that the left-hand-side of Eq.~\eqref{eq_obs} simplifies as
	\begin{equation}
		\prod_{j=1}^{q}\left(\prod_{i=1}^{p}  v_i^{\alpha_{i,j}}\right)=\prod_{i=1}^{p}  \left(\prod_{j=1}^{q} v_i^{\alpha_{i,j}}\right)= \prod_{i=1}^{p}  \left( v_i^{\sum_{j=1}^q \alpha_{i,j}} \right) = \prod_{i=1}^{p}  \left( v_i^{\text{deg}(v_i)} \right) =  1,
	\end{equation} 
where the last equality is due to the assumption that the degree of each binary variable is even, i.e., $\text{deg}(v_i)$ is even for all $i$, and in a classical solution $v_i \in \{+1, -1\}$ for all $i = 1, \ldots, p$. On the other hand, the right hand side of Eq.~\eqref{eq_obs} is $-1$ due to the assumption that the parity of the BLCS is $-1$. Therefore, no classical solution exists for the BLCS under these conditions. 
\end{proof} 

We also observe a standard form for a BLCS that does not admit a classical solution.
\begin{lemma}\label{prop1}
		Let $A$ be a binary linear constriant system that does not admit a classical solution. Then there exists a binary linear constriant system $B$ that is homomorphic to $A$ such that the degree of each variable in $B$ is even (i.e., $\text{deg}(v_i)$ is even for all $v_i \in V(B)$) and the parity of $B$ is $-1$ (i.e., $\text{par}(B) = -1$). 
\end{lemma}
\begin{proof}
To prove this lemma, we need to use the following two operations. 

\begin{operation}\label{op_negation}
For any chosen variable $v$ in the binary linear constriant system $A$, negate the parity of all constraint equations that contain $v$.
\end{operation}
Firstly, note that this operation preserves the classical realizability of the system $A$. For suppose there existed a classical solution to the new system obtained by the Operation~\ref{op_negation}. By changing the value of the variable $v$ to its opposite value ($+1 \leftrightarrow -1$) and keeping the values of all other variables, we would obtain a classical solution to the original system. Furthermore, note that applying this operation on any variable $v$ with an odd degree would change the parity of the whole system.

\begin{operation}\label{op_split}
For any chosen variable $v$ in the binary linear constriant system $A$, replace $v$ by $v_1\cdot v_2$ and add an additional constraint equation $v_1\cdot v_2=+1$ to the binary linear constraint system.
\end{operation}
Again observe that if the original BLCS does not admit a classical solution, then this Operation~\ref{op_split} preserves the classical realizability of the system. For suppose that there existed a classical solution to the new system obtained by Operation~\ref{op_split}. We would then directly obtain a classical solution to the original system in which the value of $v$ is $+1$, which is a contradiction. Furthermore, note that upon applying Operation~\ref{op_split}, we remove variable $v$ and obtain two new variables $v_1$ and $v_2$ which obey  $\text{deg}(v_1)=\text{deg}(v_2)=\text{deg}(v)+1$.

Now, we can prove the lemma by considering the following cases: 
\begin{itemize}
	\item [0.] $\text{par}(A) = -1$ and all variables in $V(A)$ have even degree. In this case, the system $A$ is already of the required standard form, i.e., $A$ and $B$ are isomorphic.  
    \item [1.] $\text{par}(A) = -1$ and there exist a non-empty set of variables $\{v_i\}$ in $V(A)$ with odd degree. Applying Operation~\ref{op_split} on each of these variables $v_i$ with odd degree gives the required system $B$. We directly see that the parity of $B$ is still $-1$ and the degree of each variable in $V(B)$ is even.
    \item [2.] $\text{par}(A) = +1$ and there exist a non-empty set of variables $\{v_i\}$ in $V(A)$ with odd degree. Applying Operation~\ref{op_negation} on exactly one of the variables $v_i$ with odd degree changes the parity of the system to $-1$. Since Operation~\ref{op_negation} does not change the degree of any variable, we are now back to case $1$. 
    \item [3.] $\text{par}(A) = +1$ and all variables in $V(A)$ have even degree. Applying Operation~\ref{op_split} on exactly one variable results in a system corresponding to case $2$.
\end{itemize}

\end{proof}

\section{Symmetric Deterministic Extension (SDE) of Nonlocal Games}\label{app:SDE}

In this section, we begin by reviewing the general framework of nonlocal games and formally introduce the notion of the \emph{Symmetric Deterministic Extension} (SDE) of nonlocal games. In Subsection~\ref{sub_blcsg}, we then focus on a specific class of bipartite nonlocal games defined for any Binary Linear Constraint System, namely the BLCS-games~\cite{CM12,Ji13}. For this class, we present explicit protocols for constructing their SDE extensions. The resulting extended games correspond to a broad family of ECID Bell inequalities, which serve as the main tool to establish the first result in the main text. Although our protocols are designed for SDE extensions of a specific class of bipartite nonlocal games, we highlight that the ECID Bell inequalities are not restricted to the bipartite case. In Subsection~\ref{sub_GHZ}, we illustrate this by providing an explicit SDE construction for the well-known GHZ game, which naturally yields a nontrivial multipartite ECID Bell inequality.

\subsection{Nonlocal Games and Symmetric Deterministic Extension}
A Bell experiment is described as a nonlocal game $G$ played by $n$ spatially separated (non-communicating) players $P_1, \ldots, P_n$ against a referee \cite{RMPBellnonlocality}. There are finite sets $X_i$ for $i \in [n]$ where $[n] := \{1,\ldots, n\}$ called the question sets and corresponding answer sets $A_i$ for $i \in [n]$, along with a rule function $W : X_1 \times \ldots \times X_n \times A_1 \times \ldots \times A_n \rightarrow \{0,1\}$, all of which are known to the players. The rule function describes the winning condition of the game, i.e., the condition that the inputs and outputs should satisfy to win the game. The players may agree on some pre-defined strategy but are not allowed to communicate during each round of the game. In each round, the referee randomly picks a set of questions $\textbf{x} \in X_1 \times \ldots \times X_n$ where $\textbf{x} = (x_1, \ldots, x_n)$ according to some probability distribution $\pi(\textbf{x})$. The referee sends question $x_i$ to the $i$-th player who returns an answer $a_i \in A_i$. The players win the round of the game if their inputs and outputs satisfy $W(\textbf{x}, \textbf{a}) = 1$ where $\textbf{a} = (a_1, \ldots, a_n)$. 

The set of joint probabilities of the outputs given the inputs $\{p(\textbf{a} | \textbf{x})\}$ (that encodes the strategy employed by the players over all the rounds) is referred to as a behavior or correlation for the Bell scenario. Classical (local hidden variable) strategies consist of a deterministic function $X_i \rightarrow A_i$ for $i\in [n]$ and convex combinations thereof obtained by using local and shared randomness. The set of classical strategies for the Bell scenario is therefore the convex hull of local deterministic behaviors (LDBs) $\{p(\textbf{a} | \textbf{x})\}$ for which $p(\textbf{a} | \textbf{x}) \in \{0,1\}$, and is denoted $\mathcal{L}(|X_1|, |A_1|; \ldots, |X_n|, |A_n|)$ or just $\mathcal{L}$ when the scenario is clear. Non-signalling strategies are the general strategies limited only by the no-communication rule, i.e., the behaviors that obey $p(\textbf{a}_S | \textbf{x}) = p(\textbf{a}_S | \textbf{x}_S)$ where $\textbf{a}_S = (a_i)_{i \in S}$ and $\textbf{x}_S = (x_i)_{i \in S}$ for any $S \subseteq [n]$. The set of no-signalling behaviors for the Bell scenario is denoted $\mathcal{N\!S}(|X_1|, |A_1|; \ldots, |X_n|, |A_n|)$ or just $\mathcal{N\!S}$ when the scenario is clear. The quantum set $\mathcal{Q}(|X_1|, |A_1|; \ldots, |X_n|, |A_n|)$ (abbreviated to $\mathcal{Q}$) is an intermediate set of behaviors that can be achieved by performing local measurements on quantum systems. Following the standard tensor-product paradigm each player is assigned a Hilbert space $\mathcal{H}_i$ of some dimension $d_i$. The players share a quantum state (a unit vector) $|\psi \rangle \in \otimes_{i=1}^n \mathcal{H}_i$ on which they perform local measurements $\{E^{(i)}_{x_i, a_i}\}_{a_i \in A_i}$ for $x_i \in X_i$, $i \in [n]$ (here the set $\{E^{(i)}_{x_i, a_i}\}_{a_i \in A_i}$ is a POVM, i.e., $E^{(i)}_{x_i, a_i} \geq 0$ and $\sum_{a_i} E^{(i)}_{x_i, a_i} = \mathtt{I}_{d_i}$). The probabilities $p(\textbf{a} | \textbf{x})$ that constitute the quantum behavior are given by $p(\textbf{a} | \textbf{x}) = \text{Tr}\left[ \otimes_{i=1}^n E^{(i)}_{x_i, a_i} |\psi \rangle \langle \psi| \right]$. 

The goal of the players is to maximize their winning probability $\omega(G)$ given by
\begin{equation}
\omega(G) := \sum_{\textbf{x}, \textbf{a}} \pi(\textbf{x}) W(\textbf{x}, \textbf{a}) p(\textbf{a} | \textbf{x}).
\end{equation}
The winning probability depends on whether the players share classical, quantum or general no-signalling resources $\{p(\textbf{a}|\textbf{x})\}$ and the corresponding optimal values are denoted $L(G)$, $Q(G)$ and $N\!S(G)$ respectively. 

It is well-known that for many games, larger winning probabilities can be achieved using quantum resources than with classical ones. A famous example is the two-player CHSH game for which the input and output sets are binary, i.e., $X_1, X_2, A_1, A_2 = \{0,1\}$ and the winning condition reads $W(x_1, x_2, a_1, a_2) = 1$ if and only if $a_1 \oplus a_2 = x_1 \cdot x_2$ where $\oplus$ denotes addition mod $2$. When the inputs are chosen uniformly $\pi(x_1, x_2) = 1/4$ for all $x_1, x_2 \in \{0,1\}$, classical strategies can at most achieve the value $L(CHSH) = 3/4$ while a quantum strategy using the maximally entangled state of two qubits achieves the value $Q(CHSH) = (2+\sqrt{2})/4$. Two-player nonlocal games $G$ for which $L(G) < Q(G) = 1$, i.e., where the quantum winning probability equals $1$, are referred to as quantum \textit{pseudo-telepathy} games, and represent a strong form of nonlocality. Another important class of bipartite nonlocal games are the Binary Linear Constraint System games \cite{CM12, Ji13} which we will use in our constructions, which we explain in detail later in the subsection~\ref{sub_blcsg}.

In the main text, we established the central result on ECID inequalities by employing the tool of \emph{symmetric deterministic extension} (SDE). In order to give a precise formal definition of SDE, and for ease of notation in the general $n$-player setting, we denote $\textbf{X} = X_1 \times \ldots \times X_n$ and $\textbf{A} = A_1 \times \ldots \times A_n$, and we focus on nonlocal games with a uniform input distribution $\Pi_U$ over all inputs.

\begin{definition} [Symmetric Deterministic Extension]\label{def:SDLifting}
Let $G = (\textbf{X}, \textbf{A}, \Pi_U, W)$ be a nonlocal game played by $n$-players $P_1, \ldots, P_n,$ where $\textbf{X} = X_1 \times \ldots \times X_n$, $\textbf{A} = A_1 \times \ldots \times A_n$ and $\Pi_U$ is the uniform distribution over $\textbf{X}$. We say that a game $\widetilde{G} = (\widetilde{\textbf{X}}, \widetilde{\textbf{A}}, \widetilde{\Pi}_U, \widetilde{W})$ where $\widetilde{\textbf{X}} = \widetilde{X}_1 \times \ldots \times \widetilde{X}_n$, $\widetilde{\textbf{A}} = \widetilde{A}_1 \times \ldots \times \widetilde{A}_n$ is a symmetric deterministic extension of game $G$ if the following conditions hold.
\begin{enumerate}
\item \textbf{Determinism}: for any given distinguished input $\textbf{x}^* = (x_1^*, \ldots, x_n^*) \in \widetilde{\textbf{X}}$, there exists a deterministic outcomes to the input $\textbf{x}^*$ such that upon fixing the deterministic outcomes for $\textbf{x}^*$ the reduced sub-game $\widetilde{G}' = (\widetilde{\textbf{X}}', \widetilde{\textbf{A}}', \widetilde{\Pi}'_U, \widetilde{W}')$ is equivalent the original game $G$ (up to a relabeling of inputs and outputs and the addition of dummy inputs and outputs), where $\widetilde{\textbf{X}}'  = \widetilde{X}_1 \setminus \{x_1^*\} \times \ldots \times \widetilde{X}_n \setminus \{x_n^*\}$, $\widetilde{\Pi}'_U$ is the uniform distribution over $\widetilde{\textbf{X}}'$, $\widetilde{W}'$ corresponds to $\widetilde{W}$ restricted to the inputs $\widetilde{\textbf{X}}'$ output $\widetilde{\textbf{A}}'$, and the output alphabet $\widetilde{\textbf{A}}$ of $\widetilde{G}$ for inputs in $\widetilde{\textbf{X}}'$ factorizes as $\widetilde{{A}}_i \cong \widetilde{{A}}'_i \times \widetilde{{A}}^*_i,\;\forall i $ where $\widetilde{{A}}^*_i$ is fixed by the deterministic outcomes associated with $\textbf{x}^*$.

\item \textbf{Symmetry}: the above property 1. holds for every $\textbf{x}^* = (x_1^*, \ldots, x_n^*) \in \widetilde{\textbf{X}}$.
\end{enumerate}
\end{definition}

Intuitively, the SDE requires that deleting any distinguished input $\textbf{x}^*$ and fixing its deterministic outputs reduces the extended game $\widetilde{G}$ to the original game $G$ (up to a relabeling of inputs and outputs and the addition of dummy inputs and outputs). 

\subsection{Bipartite BLCS-games and their the Symmetric Deterministic Extension protocols}\label{sub_blcsg}

A procedure that constructs such an extended game $\widetilde{G}$ satisfying Definition~\ref{def:SDLifting} for a given nonlocal game $G$ will be referred to as a symmetric deterministic extension protocol for $G$. In this subsection, we present explicit protocols for a specific class of bipartite nonlocal games, known as bipartite BLCS-games in the literature~\cite{CM12, Ji13}. The resulting extended games correspond to a family of ECID Bell inequalities, which provide the main tool to establish the first result in the main text.

Every binary linear constraint system $A$, that we explained in Section.~\ref{BLCS}, can be used to formulate a two-player nonlocal game $G_A$. The set of questions to one player Alice is the set of constraint equations in the system i.e. $X_1 = C(A) = \{1, \ldots, q\}$, and the set of questions to the second player Bob is the set of variables in the system $X_2 = V(A) = \{v_1, \ldots, v_p\}$. In each round, the referee randomly sends a constraint equation to Alice and a variable to Bob, so that $\pi(x_1, x_2) = \frac{1}{p q}$. Given as input the $j$-th constraint equation, Alice outputs an assignment in $\{+1, -1\}$ to each variable in the equation such that the equation is satisfied (known as the \textit{parity} condition). Given as input a variable $v_i$, Bob outputs an assignment in $\{+1, -1\}$ to the variable. The winning condition of the game is given as follows. In case Bob's input variable $v_i$ does not belong to Alice's input constraint equation (i.e., when $\alpha_{i,j} = 0$), every pair of answers by the players is accepted. In case Bob's input variable $v_i$ belongs to Alice's input constraint equation (when $\alpha_{i,j} = 1$), the referee accepts the answers if and only if the two players' assignments agree on the variable $v_i$ (known as the \textit{consistency} condition). 

A one-to-one correspondence exists between classical solutions for the BLCS $A$ and the perfect classical strategies for the corresponding nonlocal game $G_A$. On the other hand, a correspondence also exists between operator solutions for the BLCS $A$, and perfect quantum strategies for the corresponding nonlocal game $G_A$. The perfect quantum winning strategies are obtained from the operator solution in the following way:
	\begin{itemize}
		\item [1.] Replace the binary variable $v_i,\forall i\in[p]$ by the Hermitian operators $O_i$ on a Hilbert space $\mathcal{H}$ of dimension $d$ such that $O_i^2=\mathbb{I}$.
		\item [2.] Let Alice and Bob share a maximally entangled state of local dimension $d$ and perform measurements corresponding to the observables $O_i$ on their local subsystem.

	\end{itemize}
Note that in the operator solution, if two binary variables $v_i,v_j$ appear in the same constraint equation, the corresponding binary observables $O_i,O_j$ commute.
Therefore, the observables appearing in the same constraint equation can be jointly measured by Alice when she receives as input that constraint equation. Furthermore, the operator solution satisfies the constraint equations, i.e., for the $j$-th constraint equation it holds that
\begin{equation}
			\prod_{i=1}^{p}  O_i^{\alpha_{i,j}}=\beta_j\mathbb{I},
\end{equation} 
so that the parity condition is satisfied. Finally, since Alice and Bob measure observables $O_i$ on their halves of a maximally entangled state, they are guaranteed to obtain the same answers for the binary variable, i.e., the consistency condition is satisfied.

\begin{theorem}\label{theo1}
For any two-player BLCS-game $G_A$ associated with a binary linear constraint system $A$ such that $L(G_A) < Q(G_A) = 1$, there exists a symmetric deterministic extension of $G_A$.
\end{theorem}

\begin{proof}

The condition $L(G_{A})<Q(G_{A})=1$ implies that the associated BLCS $A$ admits an operator solution but no classical solution. To prove this theorem, by Proposition~\ref{prop1}, it suffices to consider the case where the parity of the corresponding BLCS $A$ is $-1$ and the degree of each variable is even.  
Without loss of generality, assume that $A$ contains $p$ variables $V(A) = \{v_1,\ldots,v_p\}$ and $q$ constraints $C(A) = \{1,\ldots,q\}$ such that each variable has even degree and the parity of the system is $\text{par}(A)=-1$. 

Let $S_{v_i}$ denote the set of constraint equations in which the variable $v_i$ appears. We are now ready to introduce the following protocol, which constructs a new BLCS $\widetilde{A}$ with $p+\sum_{i=1}^p |S_{v_i}|+q+1$ binary variables (the corresponding variable set is $V(\widetilde{A})$) and $q+1$ constraint equations (the corresponding constraint set is $C(\widetilde{A})$). 

\begin{protocol}{Symmetric Deterministic Extension for the bipartite BLCS-game $G_{A}$ with $L(G_{A})<Q(G_{A})=1$.}
\noindent\textbf{Input:} A BLCS $A$ with $\text{par}(A)=-1$ and each variable of even degree that admits an operator solution.

\noindent\textbf{Goal:} Construct a new BLCS $\widetilde{A}$ such that its associated bipartite BLCS-game $G_{\widetilde{A}}$ is a SDE of $G_{A}$ satisfying Definition~\ref{def:SDLifting}.

\noindent\textbf{The protocol:} \begin{itemize}

\item[1.] For each variable $v_i \in V(A)$ and each constraint $j \in S_{v_i}$, define a new binary variable labelled $v_{i,j}$. Multiply the left-hand-side of each constraint equation in the set $S_{v_i} \setminus \{j\}$ by the variable $v_{i,j}$. In this way, we introduce $\sum_{i=1}^p |S_{v_i}|$ new variables into the system. Let $V' = \{v_{i,j}\}$ denote the set of these new variables.

\item[2.] Define $q+1$ additional binary variables $u_1,\ldots,u_{q+1}$, with associated set $U = \{u_1,\ldots,u_{q+1}\}$. Multiply the left-hand-side of the $i$-th constraint equation by $u_i$ for each $i \in \{1,\ldots,q\}$, and multiply the left-hand-side of each constraint equation, whose parity is $-1$, by $u_{q+1}$.

\item[3.] Add the $(q+1)$-th constraint equation to the system, given as $(\prod_{v\in V'} v)\cdot (\prod_{u\in U} u)=+1$. 

\item[4.] Denote the resulting BLCS by $\widetilde{A}$. Its variable set is $V(\widetilde{A}) = V(A) \cup V' \cup U$, and its constraint set is $C(\widetilde{A}) = C(A) \cup \{q+1\}$.
\vspace{-0.8em}
\end{itemize}
\end{protocol}
\vspace{0.5em}
Now we explain why the bipartite BLCS-game $G_{\widetilde{A}}$ is a SDE of $G_{A}$. Firstly, observe that the BLCS $\widetilde{A}$ does not admit any classical solution. This is straightforwardly seen from the fact that the parity of the system is still $\text{par}(\widetilde{A})=-1$. and the degree of each binary variable is still even (recall Lemma.~\ref{obs}). To elaborate, (i) the degree of each variable in $V(A)$ doesn't change, it's even; (ii) the degree of the variable $v_{i,j}$ in $V'$ is $|S_{v_i}|$, an even number; (iii) the degree of the variable $u_i$ for $i\in\{1,\ldots,q \}$ is $2$; and (iv) the variable $u_{q+1}$ appears in every constraint equation with $-1$ parity and in the constraint equation $q+1$, so that its degree is even.

Secondly, observe that $\widetilde{A}$ admits a operator solution and for each $j \in C(\widetilde{A})$ there is a solution that assigns pure deterministic values to all the binary variables in $j$. This can be seen by considering the following cases: 
\begin{enumerate}[(i)]
\item Consider the case where $j= q+1$. In this case, we can assign deterministic value $+1$ to all new variables $\{v_{i,j}\}$ and $\{u_1,\ldots, u_{q+1}\}$, i.e., to all the variables in $V'\cup U$. The reduced system is just the BLCS $A$, for which an operator solution is known to exist. 

\item Consider the case where $j\in[q]$ and the parity of the constraint equation $j$ is $+1$. In this case, we assign the deterministic value $+1$ to each binary variable appearing in equation $j$. Moreover, we assign the deterministic deterministic value $+1$ to all the binary variables in $V'\cup U$, except for those variables $v_{i,j}$ for which the corresponding variables $v_i$ appear in equation $j$. The reduced system obtained in this way is isomorphic to the original BLCS $A$. Indeed, in the reduced system, the variables $v_{i,j}$ that are left unfixed play the role of the variables $v_i$ in equation $j$ of the original BLCS $A$. Consequently, the $(q+1)$-th constraint in the reduced system replaces constraint $j$ in the original BLCS $A$.

\item Consider the case where $j\in[q]$ and the parity of the constraint equation $j$ is $-1$. In this case, assign the deterministic value $-1$ to the variable $u_{q+1}$ and to all $u_i$ with $i\in \{1,\ldots,q\},i\neq j$ such that the parity of constraint $i$ is $-1$.
Assign the deterministic value $+1$ to all remaining variables $u_i, \; i\in \{1,\ldots,q\}$.  In addition, assign the deterministic value $+1$ to all other variables in equation $j$, and to all variables in $V'\cup U$, except for the variables $v_{i,j}$ for which the corresponding $v_i$ appear in constraint equation $j$. By the same reasoning as in case (ii), the reduced system is isomorphic to the original BLCS $A$.
\end{enumerate}
\end{proof}

We make two remarks regarding the above protocol. First, the protocol we presented is not necessarily optimal. Especially in highly symmetric BLCS instances, one may construct extensions that use fewer variables than those required by our construction. 
Secondly, in the above Theorem~\ref{theo1} we stated the protocol only for BLCS systems $A$ such that the associated bipartite BLCS-game $G_A$ satisfies $L(G_A)<Q(G_A)=1$. A natural question is whether the property extends to more general games with $Q(G_A)<1$. The protocol itself can indeed be applied in this broader class. However, in such cases the optimal quantum winning probability $Q(G_{\widetilde{A}})$ of the SDE game $G_{\widetilde{A}}$ may not be attained by the partially deterministic quantum strategies used in the proof, and must instead be computed for each specific instance $G_{\widetilde{A}}$. Nevertheless, as shown in the proof, if the partially deterministic quantum strategy achieves a value $\omega(G_{\widetilde{A}})$ strictly larger than the classical optimum, i.e.$ L(G_{\widetilde{A}}) < \omega(G_{\widetilde{A}}) \leq Q(G_{\widetilde{A}}),$ then the ECID (\textit{Every Chosen Input Deterministic}) phenomenon still holds whenever the observed winning probability of $G_{\widetilde{A}}$ lies within the interval $[L(G_{\widetilde{A}}), \,\omega(G_{\widetilde{A}})]$. This is precisely the situation discussed in the main text for the SDE version of the CHSH game $G_2$. The general case of whether arbitrary nontrivial nonlocal games with $Q < 1$ admit an SDE protocol remains open for future work.

\subsection{Symmetric Deterministic Extension for the Mermin-GHZ Game}\label{sub_GHZ}

In the previous subsection, we presented the general SDE protocol for bipartite BLCS-games, which yields a family of ECID Bell inequalities. In this subsection, we demonstrate that these ECID Bell inequalities are not restricted to the bipartite case. As a concrete example, we provide a simple geometric interpretation of the SDE for the well-known three-player Mermin-GHZ game~\cite{GHZ89,GHSZ90,Mermin1990}.

\begin{example}\label{eg:GHZ}
The symmetric deterministic extension of the cube game $G_{\text{cube}}$ (a geometric variant of the three-player Mermin-GHZ game) yields a new nonlocal game $G_{\widetilde{\text{cube}}}$ such that 
$L(G_{\widetilde{\text{cube}}}) < Q(G_{\widetilde{\text{cube}}}) = 1$ 
and $G_{\widetilde{\text{cube}}}$ satisfies Definition~\ref{def:SDLifting}.
\end{example}

\begin{proof}
The three-player Mermin-GHZ game $G_{GHZ}$ is a binary-input, binary-outcome nonlocal game that admits a perfect quantum winning strategy but no perfect classical winning strategy. The game is defined as follows. In each round, the referee uniformly at random chooses an input bit string $\bm{x}=(x_1,x_2,x_3)$ of even Hamming weight, i.e., $|\bm{x}|=\sum_{i=1}^3 x_i \equiv 0 \pmod{2}$, and sends bit $x_i \in \{0,1\}$ to player $i \in \{1,2,3\}$. Each player responds with a bit $a_i \in \{0,1\}$, and the players (who are not allowed to communicate) win the game if and only if the Hamming weight of the output string $\bm{a}=(a_1,a_2,a_3)$ equals $|\bm{x}|/2$. The maximum classical winning probability is $L(G_{GHZ}) = \frac{3}{4}$. In contrast, there exists a perfect quantum strategy achieving winning probability $Q(G_{GHZ}) = 1$, where the players share a three-qubit GHZ state, and player $i$ measures $\sigma_x$ on input $x_i=0$ and $\sigma_y$ on input $x_i=1$ to return answers.

We now formulate a geometric variant of the Mermin-GHZ game, which we call the \textit{cube} game $G_{\text{cube}}$, in the spirit of~\cite{HFSD18}. In this game, the three players, Alice, Bob, and Charlie, receive inputs that correspond to the square faces of a cube (see Fig.~\ref{fig:cube} for an illustration). Label the coordinates of the cube by $(x_1, x_2, x_3)$ with $x_1, x_2, x_3 \in \{0,1\}$. The input $x_1=0$ to Alice corresponds to the left face of the cube, while $x_1=1$ corresponds to the right face. Similarly, Bob’s input $x_2=0$ corresponds to the front face and $x_2=1$ to the back face. Analogously, Charlie’s input $x_3=0$ corresponds to the top face and $x_3=1$ to the bottom face. For each input triple $(x_1, x_2, x_3)$, the players must output binary assignments in $\{+1,-1\}$ to the vertices of the corresponding face, where the product of assignments for each face must satisfy prescribed parity constraints, namely $\text{par}(x_i=0)=+1$ and $\text{par}(x_i=1)=-1$ for all $i$. The players win the game if their outputs satisfy the consistency condition that, at the intersection vertex with coordinates $(x_1,x_2,x_3)$ (the common vertex of the three chosen faces), the product of three players' assigned values is $+1$.

\begin{figure}[H]
    \centering
    \includegraphics[width=1\linewidth]{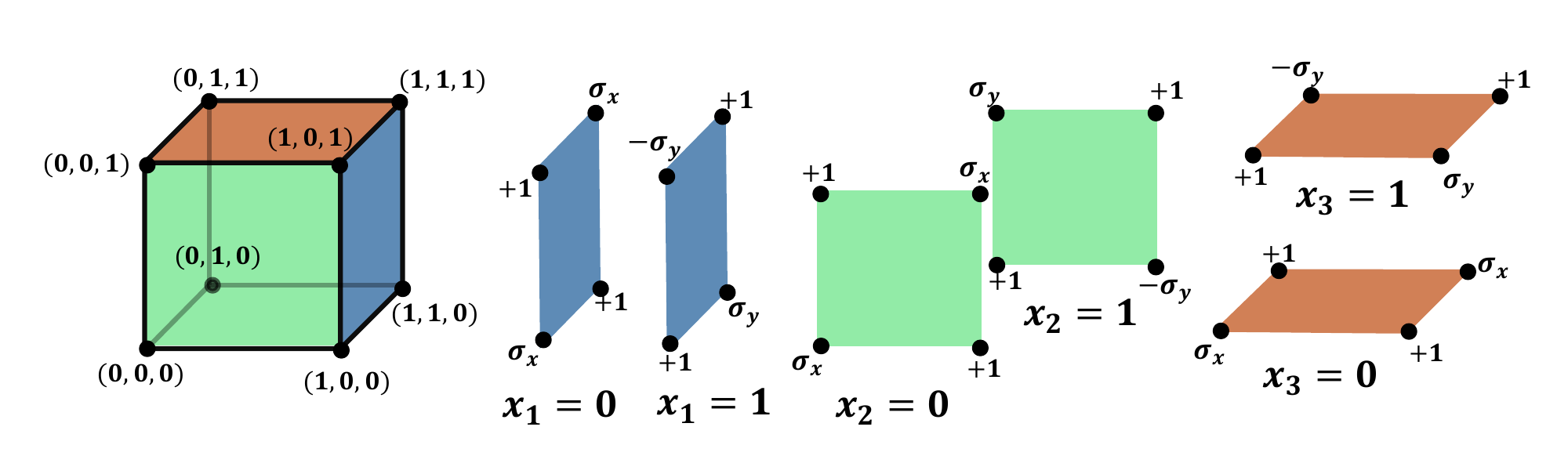}
    \caption{The graph illustration of the cube game $G_{\text{cube}}$ and its perfect quantum winning strategy. For input $x_1 = 0$, Alice returns answer $+1$ for the vertices $(0,0,1)$ and $(0,1,0)$ and returns the outcome obtained by measuring $\sigma_x$ on her qubit for the vertices $(0,1,1)$ and $(0,0,0)$. Similarly for input $x_1 = 1$, Alice returns answer $+1$ for vertices $(1,1,1)$ and $(1,0,0)$ and returns the outcome obtained by measuring $\sigma_y$ (respectively $-\sigma_y$) on her qubit for the vertex $(1,1,0)$ (respectively $(1,0,1)$). Similar strategies exist for Bob and Charlie as indicated in the green and brown faces. }
    \label{fig:cube}
\end{figure}

It is readily seen that no classical strategy exists for non-communicating players to win the game $G_{\text{cube}}$. In particular, the consistency condition requires that the product of the assignments of all players to all the vertices is $+1$ while the parity condition requires that this product is $-1$. This implies that when the parity condition is definitively imposed, the consistency condition cannot be satisfied for at least one vertex, so that the classical winning probability for this variant of the cube game is at most $ 1 - 1/8 = \frac{7}{8}$. The classical strategy that achieves $7/8$ is one where Alice, Bob and Charlie assign the value $+1$ to all vertices except the vertex $(1,1,1)$ for which they all assign the value $-1$. This strategy satisfies all the parity constraints while the consistency constraint is satisfied for all vertices other than $(1,1,1)$. Thus the maximum classical winning probability is $L(G_{\text{cube}}) =7/8.$

On the other hand, a perfect quantum winning strategy achieving $Q(G_{\text{cube}}) = 1$ exists for the game that is inherited from the perfect quantum strategy for the Mermin-GHZ game. The strategy involves the players sharing a three-qubit GHZ state $|\psi_{GHZ} \rangle = \frac{1}{\sqrt{2}} \left(|000 \rangle + | 111 \rangle \right)$ and the players return the answers obtained by performing the measurements shown in Fig. \ref{fig:cube}. To elaborate, for instance, for input $x_1 = 0$, Alice returns $+1$ for the vertices $(0,0,1)$ and $(0,1,0)$ and returns the outcome of the measurement $\sigma_x$ on her qubit for the vertices $(0,1,1)$ and $(0,0,0)$. Similarly for input $x_1 = 1$, Alice returns $+1$ for the vertices $(1,0,0)$ and $(1,1,1)$, and returns the outcome of the measurement $\sigma_y$ (respectively $-\sigma_y$) on her qubit for the vertex $(1,1,0)$ (respectively $(1,0,1)$). Similar strategies for Bob and Charlie are as shown in the Fig. \ref{fig:cube}. The parity conditions are directly seen to be satisfied for each face since evidently
\begin{equation}\label{eq:par-GHZ}
\begin{split}
    &\quad \langle \psi_{GHZ}| \left[\sigma_x \cdot \sigma_x \cdot (+1) \cdot (+1) \right] \otimes \mathbf{1} \otimes \mathbf{1} | \psi_{GHZ} \rangle = +1; \quad \langle \psi_{GHZ} | \left[(+1) \cdot (-\sigma_y) \cdot (+1) \cdot \sigma_y \right] \otimes \mathbf{1} \otimes \mathbf{1} | \psi_{GHZ} \rangle = -1,\\
    &\quad \langle \psi_{GHZ} | \mathbf{1} \otimes \left[(+1) \cdot \sigma_x \cdot (+1) \cdot \sigma_x \right] \otimes \mathbf{1} | \psi_{GHZ} \rangle = +1; \quad \langle \psi_{GHZ} | \mathbf{1} \otimes \left[(+1) \cdot (-\sigma_y) \cdot (+1) \cdot \sigma_y \right] \otimes \mathbf{1} | \psi_{GHZ} \rangle = -1,\\
    &\quad \langle \psi_{GHZ} | \mathbf{1} \otimes \mathbf{1} \otimes \left[(+1) \cdot \sigma_x \cdot (+1) \cdot \sigma_x \right] | \psi_{GHZ} \rangle = +1; \quad \langle \psi_{GHZ} | \mathbf{1} \otimes \mathbf{1} \otimes \left[(\sigma_y) \cdot (+1) \cdot (-\sigma_y) \cdot (+1) \right] | \psi_{GHZ} \rangle = -1. \\
\end{split}    
\end{equation}
The consistency conditions are also satisfied for each vertex since
\begin{equation}\label{eq:cons-GHZ}
\begin{split}
    &\quad \langle \psi_{GHZ}| \sigma_x \otimes \sigma_y \otimes (-\sigma_y)| \psi_{GHZ} \rangle = +1, \\
    &\quad \langle \psi_{GHZ} | (-\sigma_y) \otimes \sigma_x \otimes \sigma_y | \psi_{GHZ} \rangle = +1,\\
    &\quad \langle \psi_{GHZ} | \sigma_y \otimes (-\sigma_y) \otimes \sigma_x | \psi_{GHZ} \rangle = +1,\\
    &\quad \langle \psi_{GHZ} | \sigma_x \otimes \sigma_x \otimes \sigma_x | \psi_{GHZ} \rangle = +1. \\
\end{split}
\end{equation}

Formulating the Mermin-GHZ game in this geometric manner allows us to define a symmetric deterministic extension nonlocal game $G_{\widetilde{\text{cube}}}$. The game $G_{\widetilde{\text{cube}}}$ is illustrated in Fig.~\ref{fig:3by3grid}, where the vertices are labeled by triples $(x_1,x_2,x_3)$ with $x_i \in \{0,1,2\}$ for $i=1,2,3$. 

In $G_{\widetilde{\text{cube}}}$, each player $i$ receives an input $x_i \in \{0,1,2\}$. Alice’s input $x_1$ specifies the face consisting of the nine vertices $\{(x_1,x_2,x_3)\mid x_2,x_3\in\{0,1,2\}\}$. Similarly, Bob’s input $x_2$ determines the face $\{(x_1,x_2,x_3)\mid x_1,x_3\in\{0,1,2\}\}$, and Charlie’s input $x_3$ determines the face $\{(x_1,x_2,x_3)\mid x_1,x_2\in\{0,1,2\}\}$. 
For each input, the players must assign binary values in $\{+1,-1\}$ to the nine vertices of their face, subject to parity constraints that $\text{par}(x_i = 0) = +1$, $\text{par}(x_i = 1) = -1$ and $\text{par}(x_i = 2) = +1$. Specifically, for each player $i$, the product of the assigned values on the nine vertices of the face corresponding to $x_i$ must equal $\text{par}(x_i)$. Lastly, the players win the game if their assignments satisfy the consistency condition that, at the intersection vertex $(x_1,x_2,x_3)$ common to all three faces, the product of their three assignment is $+1$.

It is readily seen that no classical strategy exists for non-communicating players to win $G_{\widetilde{\text{cube}}}$. As before, the consistency condition imposes that the product of the assignments of all players to all vertices is $+1$ while the parity constraint imposes that this product is $-1$. This implies that when the parity constraint is imposed on all players, the consistency condition cannot be satisfied for at least one vertex, so that the classical winning probability for this SDE cube game is  $L(G_{\widetilde{\text{cube}}}) = 1 - 1/3^3 = \frac{26}{27}$. The classical strategy that achieves $L(G_{\widetilde{\text{cube}}})$ is one wherein Alice, Bob and Charlie assign the value $+1$ to all vertices except the vertex $(1,1,1)$ for which they all assign the value $-1$. This strategy satisfies all the parity constraints while the consistency constraint is satisfied for all vertices other than $(1,1,1)$.

On the other hand, a perfect quantum winning strategy achieving $Q(G_{\widetilde{\text{cube}}}) = 1$ exists for $G_{\widetilde{\text{cube}}}$. The strategy involves the players sharing the three-qubit GHZ state $|\psi_{GHZ} \rangle$ and the players return the answers obtained by performing the measurements shown in Fig. \ref{fig:3by3grid}. To elaborate, for instance, for input $x_1 = 0$, Alice returns $+1$ for all the vertices other than $(0,0,0)$ and $(0,1,1)$; for the vertices $(0,0,0)$ and $(0,1,1)$ she returns the value obtained by measuring $\sigma_x$ on her qubit. Similarly for input $x_1 = 1$, she returns $+1$ for all vertices other than $(1,0,1)$ and $(1,1,0)$; for the vertex $(1,1,0)$ (respectively $(1,0,1)$) she returns the value obtained by measuring $\sigma_y$ (respectively $-\sigma_y$) on her qubit. For input $x_1 = 2$, she returns $+1$ for all the nine vertices. Similar strategies for Bob and Charlie are as shown in the Fig. \ref{fig:3by3grid}. The parity conditions are directly satisfied as before. The consistency conditions are also satisfied, they hold trivially for all vertices other than $(0,0,0)$, $(1,1,0)$, $(0,1,1)$ and $(1,0,1)$; and for these vertices the consistency follows from Eq.(\ref{eq:cons-GHZ}). It is also evident from the Fig. \ref{fig:3by3grid} that the strategy is isomorphic to that for the game $G_{\text{cube}}$ applied to the inputs $x_1, x_2, x_3 \in \{0,1\}$. This strategy is partially deterministic, it returns deterministic answers for the input triple $(x_1^*, x_2^*, x_3^*) = (2,2,2)$ so that no randomness can be certified for this input.

Furthermore, such a partially deterministic quantum strategy that returns deterministic answers for the input triple $(x_1^*, x_2^*, x_3^*)$ exists for any $(x_1^*, x_2^*, x_3^*) \in \{0,1,2\}^3$ as we now show by construction. Given the triple $(x_1^*,x_2^*,x_3^*)$, the following procedure generates a partially deterministic quantum strategy that achieves $Q(G_{\widetilde{\text{cube}}}) = 1$ while returning a deterministic outcome for the triple $(x_1^*,x_2^*,x_3^*)$.
        \begin{enumerate}
            \item [1.] Assign $(+1,+1,+1)$ to all vertices $(x_1, x_2, x_3)$ for which $x_1 = x_1^*$ or $x_2 = x_2^*$ or $x_3 = x_3^*$.
            \item [2.] Assign the quantum strategy for $G_{\text{cube}}$ from Fig. \ref{fig:cube} for the remaining $8$ vertices. 
            \item [3.] There are two possibilities for Alice: 
            \begin{enumerate}
            \item [(a)] The parity constraints are satisfied for all three inputs. In this case, no modification is needed.
           \item [(b)] The parity constraints for exactly two inputs $x'_1, x''_1$, are not satisfied, note that one of the $+1$ parities is necessarily satisfied by the above strategy. In this case, flip the assignments made by Alice and Bob to the vertices $(x'_1, x_2^*, x_3^*)$ and $(x''_1,x_2^*, x_3^*)$. 
           \end{enumerate}
           \item[4.] Repeat step 3 for Bob and then for Charlie. 
        \end{enumerate}
        
\begin{figure}[H]
    \centering
    \includegraphics[width=1\linewidth]{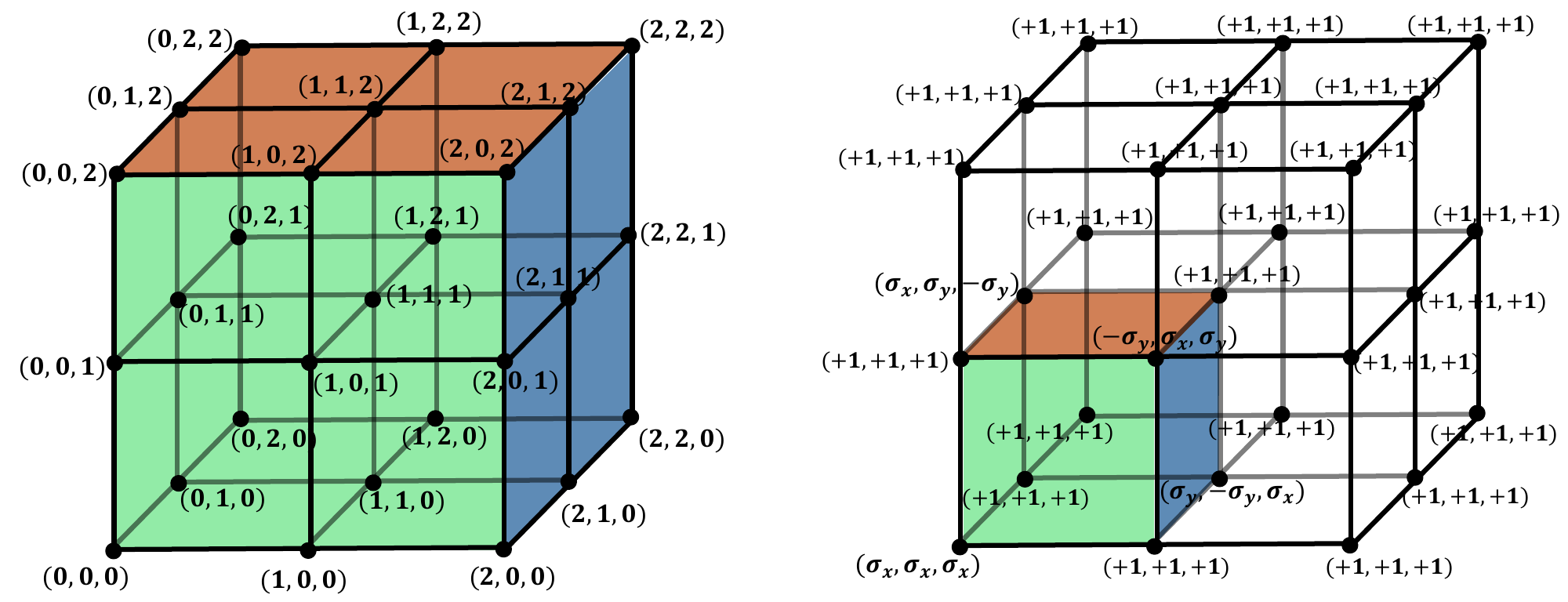}
    \caption{Left: the symmetric deterministic extension of the cube game, $G_{\widetilde{\text{cube}}}$. Right: One quantum strategy that achieves $Q(G_{\widetilde{\text{cube}}}) = 1$ while returning deterministic answers for the input triple $(x_1^*, x_2^*, x_3^*) = (2,2,2)$. The labels in brackets denote the strategies used by Alice, Bob and Charlie respectively, for instance for the vertex $(0,0,0)$ each player returns the answer obtained by measuring $\sigma_x$ on their qubit while for the vertex $(0,0,1)$ each player directly returns the value $+1$.}
    \label{fig:3by3grid}
\end{figure}

\end{proof}

\section{The Nonlocal Quantum ECID Correlation}\label{subsec:SDLMSquare}
In the previous section, we discussed how to apply the SDE construction to a family of nonlocal games to obtain ECID Bell inequalities. In each such case, every quantum strategy that achieves the maximum value is deterministic for some specific input $\bm{x}^*$, so there exist $|\textbf{X}|$ different correlations, each partly deterministic for one specific input while maximally violating the same inequality.
In this section, we introduce a stronger notion of the \textit{Every Chosen Input Deterministic} (ECID) phenomenon in quantum mechanics, namely the \emph{nonlocal quantum ECID correlation}. This refers to a single nonlocal quantum correlation that contains no certifiable randomness for any chosen $\bm{x}^*$.

The ECID correlation we consider is based on the Magic Square game $G_3$ and its SDE extension game $G_4$, as defined in the main text. Their hypergraph representations of these two games are shown in Fig.~\ref{MS}, where each vertex represents a binary variable in $\{\pm 1\}$, each hyperedge represents a constraint equation, solid edges indicate parity $+1$, and dashed edges indicate parity $-1$. 
Furthermore, it is straightforward to see that, given any valid vertex assignment (satisfying the parity constraints) for a fixed pair of row and column in the right-hand-side hypergraph, the corresponding reduced system is isomorphic to the system represented by the left-hand-side hypergraph (see Sec.~\ref{BLCS} for the notion of isomorphic BLCSs).

 \begin{figure}[H]
    \centering
    \includegraphics[width=0.5\linewidth]{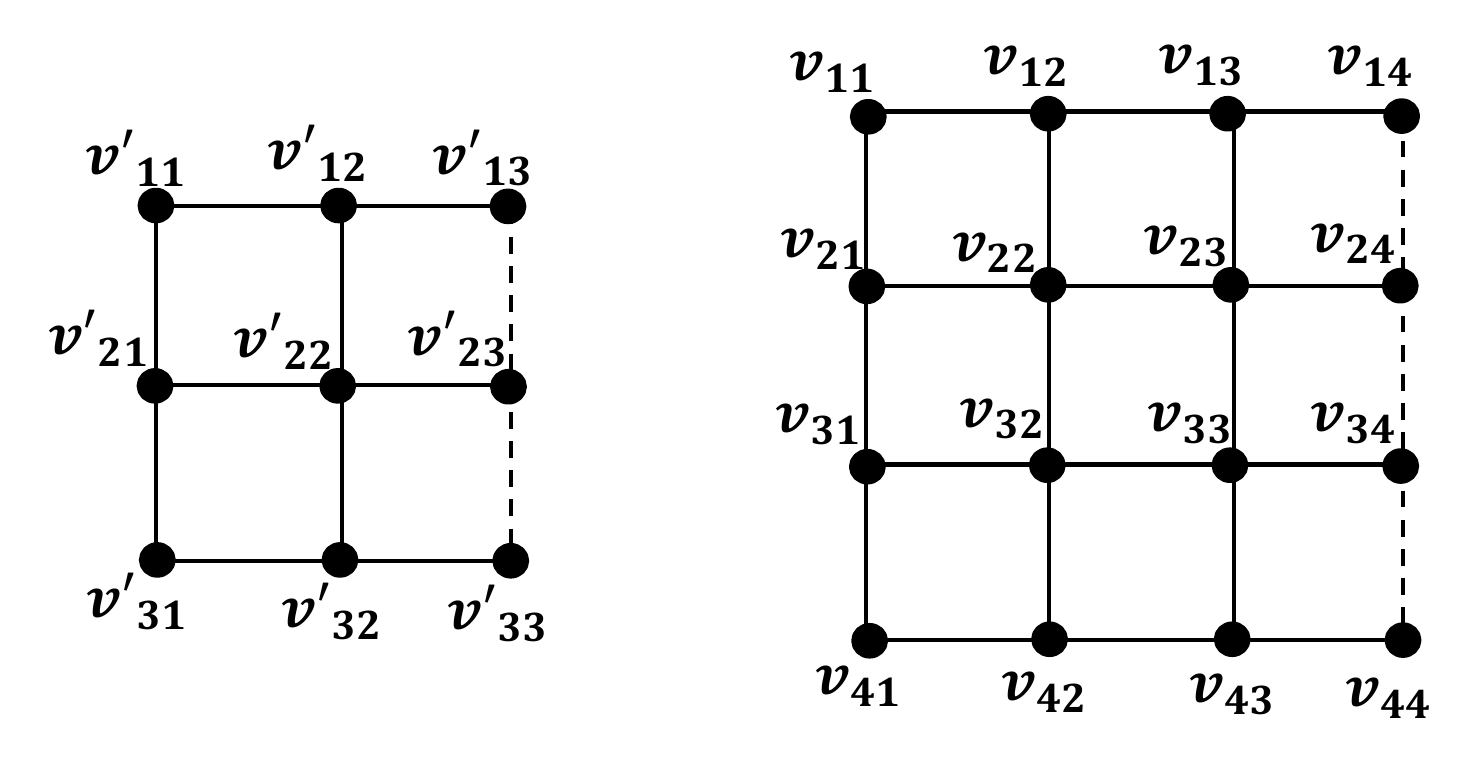}
    \caption{The hypergraph representations for the Magic Square game $G_3$ and its SDE game $G_4$. The solid lines correspond to constraint equations with parity $+1$ and the dashed lines represent constraint equations with parity $-1$.}
    \label{MS}
\end{figure}

As explained in the main text, the SDE game $G_4$ is played as follows. In each round, the referee uniformly at random selects one row $r_x,x\in\{1,\ldots,4\}$ for Alice that $r_x=\{v_{x1},v_{x2},v_{x3},v_{x4}\}$ , and one column $c_y$ for Bob that $c_y=\{v_{1y},v_{2y},v_{3y},v_{4y}\}$ (see the right of Fig.~\ref{MS} for graph illustration). 
Alice and Bob must output assignments in $\{\pm 1\}$ for all four vertices of their respective row or column, subject to the parity condition.  Denoting Alice’s output by $a$, and Bob’s output by $b$. Then $a,b\in\{(+1+1+1+1),\ldots,(-1-1+1+1)\}$ (the 8 parity $+1$ strings) except in the case when $y=4$, for which the output set is $\{(+1+1+1-1),\ldots,(-1-1-1+1)\}$ (the 8 parity $-1$ strings).
The winning condition is that Alice and Bob's assignments for the common (intersecting) vertex to their input hyperedges be the same (the consistency condition) \cite{CM12, Ji13}. 
Let $a(v_{xy})$ (similarly $b(v_{xy})$) denote Alice's (similarly Bob's) answer for the intersection vertex $v_{xy}$ of row $r_x$ and column $c_y$, the winning consistency condition of the game is then written as
\begin{equation}
    W(a,b,x,y)=
    \begin{cases} 1& \text{if } a(v_{xy})=b(v_{xy});  \\ 
    0 & \text{otherwise.} \\ 
    \end{cases}
\end{equation}

We claim that the following correlation $\{ p_{G_4}(a,b|x,y) \}$ is a nonlocal quantum ECID correlation:
\begin{equation}
    p_{G_4}(a,b|x,y)=\begin{cases} \frac{1}{32}& \text{when } W(a,b,x,y)=1;  \\ 
    0 & \text{when } W(a,b,x,y)=0. \\ 
    \end{cases}
\end{equation}

First, this correlation is nonlocal, since it wins the SDE game $G_4$ with unit probability while the optimal classical winning probability is $15/16$, strictly smaller than $1$.
Next, we explain why this correlation exhibits the ECID phenomenon and is quantum-realizable. Fix any input pair $(\hat x,\hat y)$ in the SDE game $G_4$. There exists a family of $32$ perfect quantum strategies, each of which:
(i) assigns deterministic outputs at the chosen inputs $(\hat x,\hat y)$ (consistent with the parity and consistency constraints), and
(ii) on the remaining inputs implements the perfect quantum strategy for the reduced game that is isomorphic to the perfect quantum strategy for the Magic Square game $G_3$ (see Fig.~\ref{lift_PM_1} and Fig.~\ref{SDLMsqdecomposition}).
Moreover, the uniform convex combination of the corresponding $32$ correlations reproduces exactly the correlation $\{ p_{G_4}(a,b|x,y) \}$. 
In other words, for every chosen $(\hat x,\hat y)$, the same global correlation $\{ p_{G_4}(a,b|x,y) \}$ admits a decomposition into perfect quantum strategies that are deterministic on $(\hat x,\hat y)$. We detail this construction below.

Consider as an instance the case $(\hat x=1,\hat y=1)$ for which one such strategy is shown in Fig.~\ref{lift_PM_1}, we denote the corresponding correlation by $\left\{ p_{G_4}^{r_1,c_1,1} (a,b|x,y) \right\}$.
\begin{figure}[htbp]
    \centering
\includegraphics[width=0.2\linewidth]{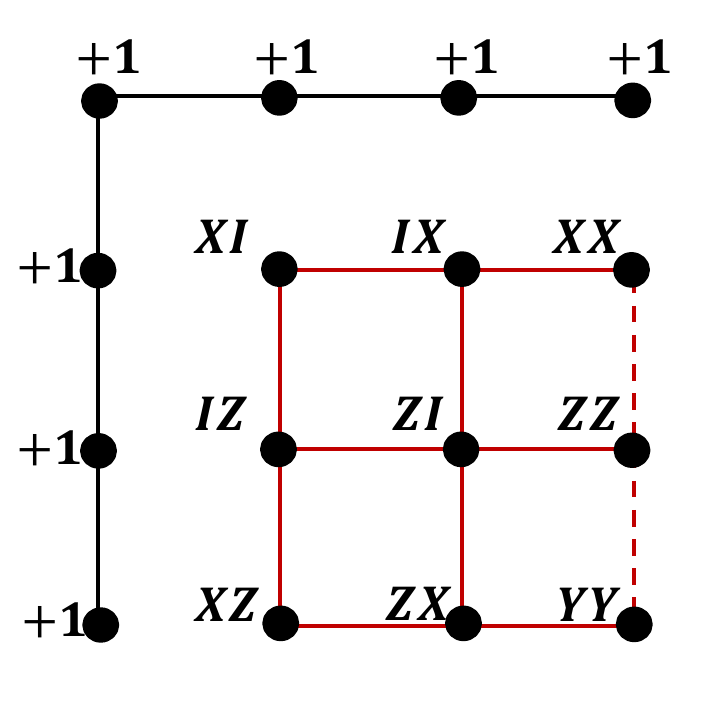}
    \caption{One perfect quantum strategy for $G_4$ that corresponds to the quantum correlation $\left\{ p_{G_4}^{r_1,c_1,1} \right\}$. The deterministic outcomes (consistent with the parity and consistency constraints) are assigned to inputs $(\hat x=1,\hat y=1)$ (hyperedges $r_1,c_1$) and a perfect quantum strategy for the $G_3$ is assigned to the reduced system (the red portion).}
    \label{lift_PM_1}
\end{figure}
We know that the nonlocal quantum correlation $\{p_{G_3}(a,b|x,y)\}$ for the Magic Square game $G_3$ has $p_{G_3}(a,b|x,y)=1/8$ when the winning condition holds and $p_{G_3}(a,b|x,y)=0$ when the winning condition doesn't hold.
Thus the correlation $\left\{ p_{G_4}^{r_1,c_1,1} (a,b|x,y) \right\}$ for the strategy shown in Fig.~\ref{lift_PM_1} is
\begin{equation}
    p_{G_4}^{r_1,c_1,1} (a,b|x,y)=\begin{cases} 1 & \text{when } x=1,y=1 \text{ and } a=(+1+1+1+1), b=(+1+1+1+1);  \\ 
    \frac{1}{4}& \text{when } \big(x=1,y\neq 1 \text{ and } b(v_{xy})=+1 \big)\text{ or } \big(x\neq 1,y=1 \text{ and } a(v_{xy})=+1 \big);  \\ 
    \frac{1}{8}& \text{when } x\neq 1,y\neq 1 \text{ and } a(v_{xy})=b(v_{xy});  \\ 
    0 & \text{otherwise }. \\ 
    \end{cases}
\end{equation}

Furthermore, for fixed $(\hat x=1,\hat y=1)$, i.e., hyperedges $r_1,c_1$, there are $\binom{6}{1} + \binom{6}{3} + \binom{6}{5} = 32$ subsystems corresponding to the Magic Square system, and for each of them, we have a perfect quantum winning strategy similar to Fig.~\ref{lift_PM_1} that assigns deterministic values (consistent with the parity and consistency constraints) to hyperedges $r_1,c_1$. In Fig.~\ref{SDLMsqdecomposition}, we list all $32$ such perfect quantum winning strategies. Denoting the corresponding nonlocal quantum correlations by $\left\{ p_{G_4}^{r_1,c_1,i} (a,b|x,y) \right\}$ that $i\in\{1,\ldots,32\}$, the correlation associated to the $i$-th strategy is
\begin{equation}
    p_{G_4}^{r_1,c_1,i} (a,b|x,y)=\begin{cases} 1 & \text{when } x=1,y=1 \text{ and } a=a_i,b=b_i;  \\ 
    \frac{1}{4}& \text{when } \big(x=1,y\neq 1 \text{ or } x\neq 1,y=1 \big) \text{ and } a(v_{xy})=b(v_{xy});  \\ 
    \frac{1}{8}& \text{when } x\neq 1,y\neq 1 \text{ and } a(v_{xy})=b(v_{xy});  \\ 
    0 & \text{otherwise }. \\ 
    \end{cases}
\end{equation}
where $a_i,b_i$ are deterministic values assigned to hyperedges $r_1,c_1$ in the $i$-th strategy. The behavior$\{ p_{G_4}(a,b|x,y) \}$ is then seen to admit the convex decomposition 
\begin{equation}
   \left\{ p_{G_4}(a,b|x,y) \right\} = \frac{1}{32} \sum_{i=1}^{32}\left\{ p_{G_4}^{r_1,c_1,i} (a,b|x,y) \right\}.
\end{equation}

 \begin{figure}[H]
    \centering
    \includegraphics[width=0.55\linewidth]{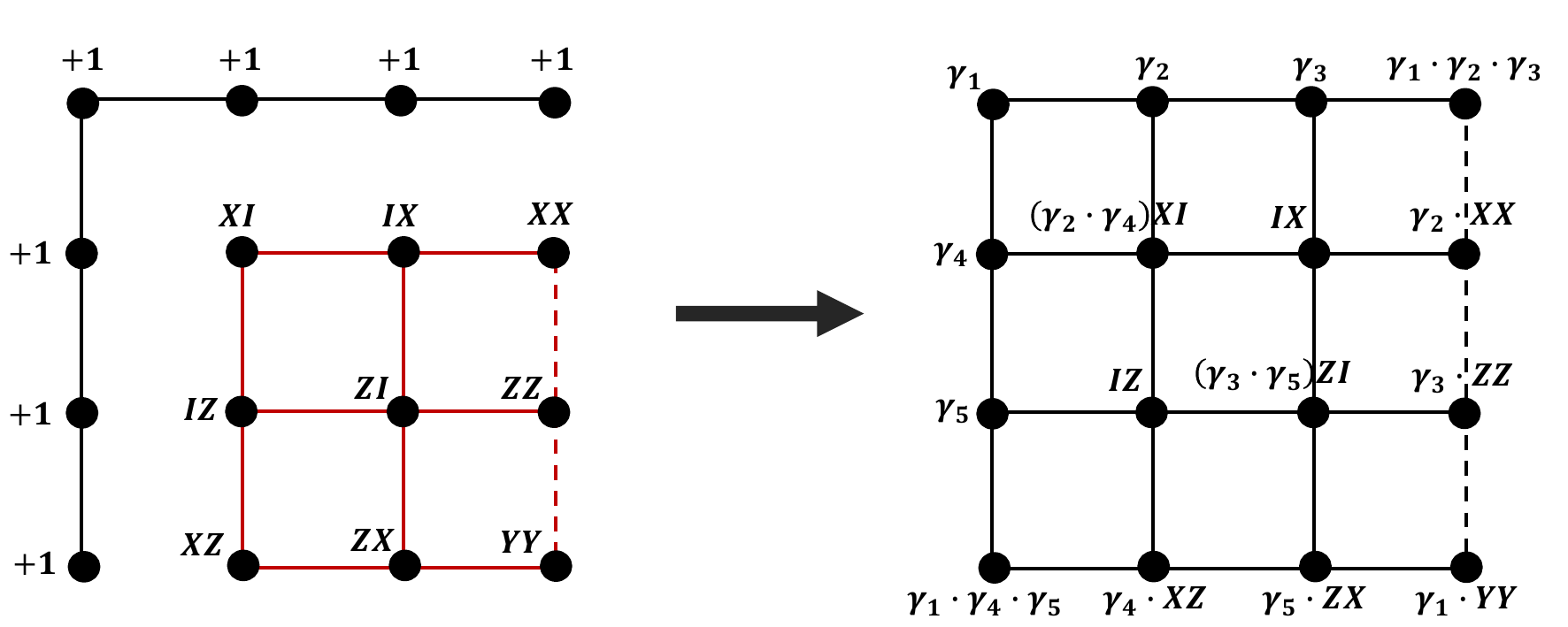}
    \caption{On the left is represented the partially deterministic strategy that assigns deterministic outcome $+1$ to all variables for input hyperedges $r_1, c_1$, the corresponding behavior $\left \{ p_{G_4}^{r_1,c_1,1} (a,b|x,y)  \right \}$ is described in the text. On the right is represented the set of $32$  partially deterministic strategies that output deterministic values for the variables in $r_1, c_1$. These $32$ strategies correspond to the outputs $\gamma_1, \gamma_2, \gamma_3, \gamma_4, \gamma_5 \in \{+1, -1\}$ for the vertices $v_{11}, v_{12}, v_{13}, v_{21}, v_{31}$ respectively. Note that the vertex $v_{14}$ is assigned value $\gamma_1 \cdot \gamma_2 \cdot \gamma_3$ to satisfy the parity $+1$ condition for row $1$. Similarly, the vertex $v_{41}$ is assigned value $\gamma_1 \cdot \gamma_4 \cdot \gamma_5$ to satisfy the parity $+1$ condition for column $1$. The outcomes for the other vertices in the hypergraph are assigned from the outputs of the labelled quantum measurements as shown. It is readily checked that the parity condition is satisfied for all rows (equal to $+1$) and all columns (equal to $+1$ for all column except the fourth column for which it is $-1$).}
    \label{SDLMsqdecomposition}
\end{figure}

From the Fig.\ref{SDLMsqdecomposition}, it is evident that a similar convex decomposition of $\left\{ p_{G_4}(a,b|x,y) \right\}$ holds for each choice of $r_{\hat x}, c_{\hat y}$ for any $\hat x, \hat y \in \{1,\ldots, 4\}$. Nevertheless, let us formalise this in the remaining portion of this section.

Let $S^{r_m,c_n,i}$ denote the hypergraph representation of the $i$-th strategy which assigns purely deterministic values to the outcomes for hyperedges $r_m,c_n$. The $32$ strategies $S^{r_1,c_1,i},i\in\{1,\ldots,32\}$ are shown in Fig.~\ref{SDLMsqdecomposition} on the right, which assign deterministic values to the outcomes for settings $r_1,c_1$ and operators corresponding to a perfect strategy for a Magic Square game $G_3$ for the reduced vertices in $\{r_2,r_3,r_4\},\{c_2,c_3,c_4\}$. We show how to construct an analogous set of $32$ strategies $S^{r_m,c_n,i}$ for any chosen pair $r_m, c_n$ for $m, n \in \{1,\ldots, 4\}$. That is, we detail the assignment (either of a deterministic value or a quantum operator) to the vertices in the hypergraph representation of strategy $S^{r_m,c_n,i}$. 

Each vertex in this hypergraph representation of strategy $S^{r_m,c_n,i}$ is uniquely identified as the intersection point of some row and column, i.e., $r_k, c_l$ for some $k, l \in \{1,\ldots, 4\}$, denoting its assignment by  $o^{S^{r_m,c_n,i}}(r_k,c_l)$, which is obtained from the assignment to $o^{S^{r_1,c_1,i}}(r_k,c_l)$ in the hypergraph strategy $S^{r_1,c_1,i}$ as
\begin{equation}
    o^{S^{r_m,c_n,i}}(r_k,c_l) :=\begin{cases}
    \text{par}(c_n)\cdot o^{S^{r_1,c_1,i}}(r_{\pi(k)},c_{\sigma(l)})) & \text{if } k=1,l=1 \text{ or } k=m,l=n;\\
    o^{S^{r_1,c_1,i}}(r_{\pi(k)},c_{\sigma(l)}) & \text{otherwise.}
    \end{cases}
\end{equation}
where $\pi(m)=1,\pi(1)=m,$ and $\pi(i)=i,\forall i\neq 1,m$, $\sigma(n)=1,\sigma(1)=n$ and $\sigma(i)=i,\forall i\neq 1,n$. And $\text{par}(c_1)=\text{par}(c_2)=\text{par}(c_3)=1$ and $\text{par}(c_4)=-1$.

Thus the behavior $\{p_{G_4}^{r_m,c_n,i}(a,b|x,y)\}$ corresponding to the hypergraph strategy $S^{r_m,c_n,i}$ for $m,n\in\{1,\ldots,4\},i\in\{1,\ldots,32\}$ is given by
\begin{equation}
    p_{G_4}^{r_m,c_n,i}(a,b|x,y) :=\begin{cases} 1 & \text{when } x=m,y=n \text{ and } a=a_{S^{r_m,c_n,i}},b=b_{S^{r_m,c_n,i}};  \\ 
    \frac{1}{4}& \text{when } \big(x=m,y\neq n \text{ or } x\neq m,y=n \big) \text{ and } a(v_{xy})=b(v_{xy});  \\ 
    \frac{1}{8}& \text{when } x\neq m,y\neq n \text{ and } a(v_{xy})=b(v_{xy});  \\ 
    0 & \text{otherwise }. \\ 
    \end{cases}
\end{equation}
where $a_{S^{r_m,c_n,i}},b_{S^{r_m,c_n,i}}$ are the deterministic values assigned to hyperedges $r_m,c_n$ in the $i$-th strategy $S^{r_m,c_n,i}$. Since the partially deterministic behaviors are obtained by simple permutations, the correlation $\{p_{G_4}(a,b|x,y)\}$ is seen to admit the convex decomposition for any $m, n \in \{1, \ldots, 4\}$
\begin{equation}
   \left\{ p_{G_4}(a,b|x,y) \right\} = \frac{1}{32} \sum_{i=1}^{32}\left\{ p_{G_4}^{r_m,c_n,i} (a,b|x,y) \right\}.
\end{equation}